\documentclass[a4paper,USenglish,cleveref,pdfa]{lipics-v2021}

\pdfoutput=1
\hideLIPIcs
\nolinenumbers

\newif\iflncs
\lncsfalse

\usepackage{preamble}
\title{Structural Parameterizations for Induced and Acyclic Matching}

\titlerunning{Structural Parameterizations for Induced and Acyclic Matching}

\author{Michael Lampis}
{Universit\'{e} Paris-Dauphine, PSL University, CNRS UMR7243, LAMSADE, Paris, France}
{michail.lampis@dauphine.fr}
{https://orcid.org/0000-0002-5791-0887}{}

\author{Manolis Vasilakis}
{Universit\'{e} Paris-Dauphine, PSL University, CNRS UMR7243, LAMSADE, Paris, France}
{emmanouil.vasilakis@dauphine.eu}
{https://orcid.org/0000-0001-6505-2977}{}

\authorrunning{M. Lampis and M. Vasilakis}

\Copyright{Michael Lampis and Manolis Vasilakis}

%TODO mandatory: Please choose ACM 2012 classifications from https://dl.acm.org/ccs/ccs_flat.cfm

%TODO mandatory; please add comma-separated list of keywords

\category{} %optional, e.g. invited paper

\funding{This work is partially supported by ANR project ANR-21-CE48-0022 (S-EX-AP-PE-AL).}

\EventEditors{Neeldhara Misra and Magnus Wahlstr\"{o}m}
\EventNoEds{2}
\EventLongTitle{18th International Symposium on Parameterized and Exact Computation (IPEC 2023)}
\EventShortTitle{IPEC 2023}
\EventAcronym{IPEC}
\EventYear{2023}
\EventDate{September 6--8, 2023}
\EventLocation{Amsterdam, The Netherlands}
\EventLogo{}
\SeriesVolume{285}
\ArticleNo{22}
\begin{document}

\maketitle

\begin{abstract}
We revisit the (structurally) parameterized complexity of \textsc{Induced Matching} and \textsc{Acyclic Matching},
two problems where we seek to find a maximum independent set of edges whose endpoints induce,
respectively, a matching and a forest.
Chaudhary and Zehavi~[WG '23, SIDMA '25] recently studied these problems parameterized by treewidth, denoted by $\mathrm{tw}$.
We resolve several of the problems left open in their work and extend their results as follows:
(i) for \textsc{Acyclic Matching}, Chaudhary and Zehavi gave an algorithm of running time
$6^{\mathrm{tw}}n^{\mathcal{O}(1)}$ and a lower bound of
$(3-\varepsilon)^{\mathrm{tw}}n^{\mathcal{O}(1)}$ (under the SETH);
we close this gap by, on the one hand giving a more careful analysis of their
algorithm showing that its complexity is actually $5^{\mathrm{tw}} n^{\mathcal{O}(1)}$,
and on the other giving a pw-SETH-based lower bound showing that this running time cannot be improved
(even for pathwidth),
(ii) for \textsc{Induced Matching} we show that their $3^{\mathrm{tw}} n^{\mathcal{O}(1)}$
algorithm is optimal under the pw-SETH (in fact improving over this for pathwidth or even for cutwidth is \emph{equivalent} to falsifying the pw-SETH)
by adapting a recent reduction for \textsc{Bounded Degree Vertex Deletion},
(iii) for both problems we give FPT algorithms with single-exponential dependence when
parameterized by clique-width and in particular for \textsc{Induced Matching}
our algorithm has running time $3^{\mathrm{cw}} n^{\mathcal{O}(1)}$,
which is optimal under the pw-SETH from our previous result.

% ------------------------------ ALTERNATIVELY ------------------------------
% We revisit the (structurally) parameterized complexity of \textsc{Induced Matching} and \textsc{Acyclic Matching},
% two problems where we seek to find a maximum independent set of edges whose endpoints induce,
% respectively, a matching and a forest.
% A recent work by Chaudhary and Zehavi~[WG '23] shows that the parameterization by the treewidth $\mathrm{tw}$
% of the input graph yields FPT algorithms of running times $3^{\mathrm{tw}} n^{\mathcal{O}(1)}$ and $6^{\mathrm{tw}} n^{\mathcal{O}(1)}$ respectively,
% and they ask whether their algorithms are optimal.
% We first show that, with some slightly more careful analysis, the latter algorithm has running time
% $5^{\mathrm{tw}} n^{\mathcal{O}(1)}$ instead.
% Our main contribution however is to show that any further improvement in any of them would refute the SETH,
% even for the parameterization by pathwidth, thus rendering both algorithms optimal.
% Moving on, we study both problems when parameterized by the clique-width $\mathrm{cw}$ of the input graph.
% Here, we develop an optimal (modulo SETH) algorithm for \textsc{Induced Matching} of running time
% $3^{\mathrm{cw}} n^{\mathcal{O}(1)}$, while for \textsc{Acyclic Matching} we develop an FPT algorithm of
% single-exponential parametric dependence by making use of the framework developed by Bergougnoux and Kant{\'e}~[TCS].

\end{abstract}

\newpage

%%%%%%%%%%%%%%%%%%%%%%%%%%%%%%%%%%%%%%%%%%%%%%%%%%%%%%%%%
% --------------------- INTRODUCTION --------------------
%%%%%%%%%%%%%%%%%%%%%%%%%%%%%%%%%%%%%%%%%%%%%%%%%%%%%%%%%
\section{Introduction}\label{sec:introduction}

Given a graph $G$, a \emph{matching} $M \subseteq E(G)$ of $G$ is a subset of its edges such that
every vertex of $G$ is incident with at most one edge in $M$.
Matchings are central in Computer Science~\cite{books/LovaszP09},
and problems involving them find numerous practical applications such as matching kidney donors to recipients
and clients to server clusters to name a few~\cite{books/ws/Manlove13}.
Even though finding a matching of maximum cardinality in general graphs
is well-known to be polynomial-time solvable~\cite{focs/MicaliV80},
further demanding that the graph induced by the vertices of the matching adheres to a specific constraint
oftentimes renders the problem NP-hard.
Recently, Chaudhary and Zehavi~\cite{wg/ChaudharyZ23a,siamdm/ChaudharyZ25} studied a plethora of such NP-hard problems
under the lens of parameterized complexity, providing a variety of algorithms and posing some open questions,
some of which we answer in this work.

We first define the problems in question.
Let $G=(V,E)$ be a graph and $M \subseteq E$ a matching of $G$, where $V_M \subseteq V$ denotes
the set of vertices of $G$ incident with edges in $M$.
We say that matching $M$ is \emph{induced} (resp., \emph{acyclic}) if $G[V_M]$ is $1$-regular (resp., acyclic).
In that case, given a graph $G$ and an integer $\ell$,
{\InducedM} (resp., \AcyclicM) asks whether $G$ has an induced (resp., acyclic) matching of size at least $\ell$.
Originally introduced as the ``risk-free'' marriage problem~\cite{ipl/StockmeyerV82},
{\InducedM} has been extensively studied due to its numerous applications (see~\cite{dam/GolumbicL00}).
This problem and its variations also have connections to various problems such as
\textsc{Maximum Feasible Subsystem}~\cite{soda/ChalermsookLN13,soda/ElbassioniRRS09}
and storylines extraction~\cite{kdd/KumarMS04}.
% maximum expanding sequence~\cite{siamcomp/BriestK11}
% It is also a subtask of finding a \emph{strong edge coloring}.
As for \AcyclicM, it was introduced by Goddard et al.~\cite{dm/GoddardHHL05}
along with some other variations of \textsc{Maximum Matching}.
Both problems are known to be NP-hard~\cite{dam/Cameron89,dm/GoddardHHL05,ipl/StockmeyerV82},
thus in this paper we examine them under the perspective of parameterized complexity.%
\footnote{We assume the reader is familiar with the basics of parameterized complexity,
as given e.g.~in~\cite{books/CyganFKLMPPS15}.}
Most relevant to us are the results concerning their complexity under structural parameterizations.
Both problems are known to be FPT by the treewidth $\tw$ of the input graph~\cite{siamdm/ChaudharyZ25,tcs/HajebiJ23,dam/MoserS09},
with the state-of-the-art being due to Chaudhary and Zehavi~\cite{siamdm/ChaudharyZ25}
who proposed algorithms of running time%
\footnote{Standard $\sO$ notation is used to suppress polynomial factors.}
$\sO(3^{\tw})$ and $\sO(6^{\tw})$ for {\InducedM} and {\AcyclicM} respectively.
Furthermore, the authors complement their algorithms with
$\sO((\sqrt{6}-\varepsilon)^{\pw})$ and $\sO((3-\varepsilon)^{\pw})$ SETH-based lower bounds respectively,
where $\pw$ denotes the pathwidth of the input graph.
Lastly, as mentioned in~\cite{algorithmica/KoblerR03}, {\InducedM} is expressible in MSO$_1$ logic,
thus it is FPT parameterized by clique-width due to standard metatheorems~\cite{mst/CourcelleMR00}.

\subparagraph{Our Contribution.}
We start with the {\InducedM} problem.
Our first result is to show the optimality of the $\sO(3^{\tw})$ algorithm of Chaudhary and Zehavi~\cite{siamdm/ChaudharyZ25}.
As a matter of fact, for the parameterizations by the pathwidth $\pw$ and the cutwidth $\ctw$ of the input graph,
we show that for all $\varepsilon>0$, obtaining an $\sO((3-\varepsilon)^{\pw})$ or an $\sO((3-\varepsilon)^{\ctw})$ algorithm is \emph{equivalent} to
falsifying the pw-SETH, a weakening of the SETH that was recently postulated by Lampis~\cite{soda/Lampis25}.
Note that this implies that an $\sO((3-\varepsilon)^{\pw})$ algorithm would falsify not just the standard form of the SETH (for $k$-CNF formulas),
but (among others) also the SETH for circuits of depth $\varepsilon n$~\cite{arxiv/Lampis24} and the Set Cover Conjecture~\cite{CyganDLMNOPSW16}.
Our reduction follows along the lines of the one for {\BDD} by Lampis and Vasilakis~\cite{toct/LampisV24}, though adapted to the machinery of the pw-SETH.
We remark that concurrently and independently to our work,
Bojikian, Chekan, and Kratsch~\cite{arxiv/BojikianCK25} recently showed that an $\sO((3-\varepsilon)^{\ctw})$ algorithm for {\InducedM}
would falsify the SETH.
In comparison, our result is stronger in two aspects,
as not only do we base the hardness on the weaker pw-SETH, but we in fact show the equivalence to it.
Moving on, we consider the parameterization by the clique-width $\cw$ and develop an (optimal due to our previous result)
$\sO(3^{\cw})$ algorithm based on dynamic programming.
Although obtaining an algorithm of running time $\sO(6^{\cw})$ is quite straightforward,
improving this to our obtained $\sO(3^{\cw})$ running time requires some extra machinery.
In particular, to bypass the bottleneck of the computation time in the union nodes,
we employ ideas from~\cite{mfcs/BodlaenderLRV10,esa/RooijBR09} and consider \emph{merged states} while solving the
counting version of the problem.
Doing so allows us to subsequently use the FFT technique introduced by Cygan and Pilipczuk~\cite{tcs/CyganP10}
and populate the whole table in quasilinear time.
Next we move towards the {\AcyclicM} problem.
We first provide a more careful analysis of the recent algorithm by Chaudhary and Zehavi~\cite{siamdm/ChaudharyZ25}
and show that in fact it runs in time $\sO(5^{\tw})$.
Our main result however is to show that improving over this, even for the parameterization by pathwidth,
results in the pw-SETH failing.
Finally, we consider {\AcyclicM} parameterized by clique-width.
Bergougnoux and Kant\'e~\cite{tcs/BergougnouxK19} developed a framework based on the rank-based approach of
Bodlaender, Cygan, Kratsch, and Nederlof~\cite{iandc/BodlaenderCKN15}
to deal with connectivity problems in this setting and developed algorithms with single-exponential dependence on the parameter
for various problems, including \FVS.
Adapting the latter algorithm allows us to obtain an $2^{\bO(\cw)} n^{\bO(1)}$ algorithm for \AcyclicM,
which is optimal under the ETH.

\subparagraph{Related Work.}
Both {\InducedM} and {\AcyclicM} have been extensively studied with respect to their polynomial-time
solvability in specific graph classes.
We mention in passing that both problems remain NP-hard even in very restricted graph classes;
for the former, such results are known for bipartite graphs of maximum degree $3$~\cite{ipl/Lozin02}
and planar graphs of maximum degree $4$~\cite{siamdm/KoS03},
while for the latter for planar bipartite graphs of maximum degree $3$~\cite{tcs/HajebiJ23}.
For further results in specific graph classes we refer to~\cite{tcs/BrandstadtES07,algorithmica/BrandstadtH08,dam/Cameron89,dm/CameronST03,dam/Chang03,dmtcs/EkimD13,dam/GolumbicL00,algorithmica/HabibM20,algorithmica/KlemzR22,siamdm/KoS03,algorithmica/KoblerR03,dam/KrishnamurthyS12,ipl/Lozin02,wg/Zito99}
for {\InducedM} and to~\cite{dam/BasteR18,tcs/HajebiJ23,tcs/PandaC23,dmaa/PandaP12} for \AcyclicM.
Exact exponential-time algorithms for {\InducedM} have been proposed in the literature~\cite{ol/ChangCH15,iandc/XiaoT17},
while both problems have also been studied with respect to their (in)approximability;
see~\cite{disopt/BasteFR20,soda/ChalermsookLN13,focs/ChalermsookLN13,tcs/ChlebikC06,iandc/ChlebikC08,jda/DuckworthMZ05,soda/ElbassioniRRS09,tcs/FurstLR18,dam/LinMV18,disopt/OrlovichFGZ08,tcs/Rautenbach15}
and~\cite{dm/BasteFR22,dam/BasteR18,anor/FurstR19,tcs/PandaC23} respectively.

From the viewpoint of Parameterized Complexity,
both {\InducedM} and {\AcyclicM} are W[1]-hard parameterized by the natural parameter $\ell$ even in bipartite graphs,
yet become FPT for various graph classes, including planar graphs~\cite{tcs/DabrowskiDL13,tcs/HajebiJ23,dam/MoserS09};
as a matter of fact, in case of {\InducedM} in planar graphs, various kernelization algorithms have been developed~\cite{dam/ErmanKKW10,jcss/KanjPSX11}.
Recently, Chaudhary and Zehavi~\cite{jcss/ChaudharyZ25} also provided parameterized inapproximability results.
Below guarantee parameterizations have also been studied~\cite{jcss/ChaudharyZ25,stacs/Koana23,jda/MoserT09,tcs/XiaoK20}.

Regarding structural parameterizations, both problems are FPT parameterized by the treewidth of the input graph, denoted by $\tw$.
For \InducedM, Moser and Sikdar~\cite{dam/MoserS09} developed a $\sO(4^{\tw})$-time DP algorithm
which was subsequently improved to $\sO(3^{\tw})$ by Chaudhary and Zehavi~\cite{siamdm/ChaudharyZ25}
who also showed a $\sO((\sqrt{6}-\varepsilon)^{\pw})$ lower bound under the SETH ($\pw$ denotes the pathwidth).
As for \AcyclicM, Hajebi and Javadi~\cite{tcs/HajebiJ23} noticed that the problem is expressible in MSO$_2$,
thus it is FPT parameterized by treewidth due to Courcelle's theorem~\cite{iandc/Courcelle90};
Chaudhary and Zehavi~\cite{siamdm/ChaudharyZ25} later provided an explicit algorithm of complexity $\sO(6^{\tw})$
as well as a $\sO((3-\varepsilon)^{\pw})$ lower bound under the SETH.
For structural parameters apart from treewidth,
it is known that {\InducedM} is expressible in MSO$_1$ logic~\cite{algorithmica/KoblerR03},
thus FPT parameterized by clique-width due to standard metatheorems~\cite{mst/CourcelleMR00},
and {\AcyclicM} is FPT parameterized by modular-width~\cite{tcs/HajebiJ23},
while neither problem admits (under standard assumptions) a polynomial kernel parameterized by
either the vertex cover number or the distance to clique of the input graph~\cite{jcss/ChaudharyZ25,tcs/GomesMPSS23}.

%%%%%%%%%%%%%%%%%%%%%%%%%%%%%%%%%%%%%%%%%%%%%%%%%%%%%%%%%
% -------------------- PRELIMINARIES --------------------
%%%%%%%%%%%%%%%%%%%%%%%%%%%%%%%%%%%%%%%%%%%%%%%%%%%%%%%%%
\section{Preliminaries}\label{sec:preliminaries}
Throughout the paper we use standard graph notations~\cite{books/Diestel17},
and we assume familiarity with the basic notions of parameterized complexity~\cite{books/CyganFKLMPPS15}.
All graphs considered are undirected without loops.
Let $G = (V, E)$ be a graph.
Then, $\cc(G)$ denotes the set of its connected components.
Given a subset of its vertices $S \subseteq V$, $G[S]$ denotes the subgraph induced by $S$ while $G - S$ denotes $G[V \setminus S]$.
For a weight function $\wc$ on the vertices of $G$ and a subset $S \subseteq V$,
we denote by $\wc(S)$ the sum of the weights of all vertices of $S$.
Given a subset of its edges $M \subseteq E$, $V_M$ denotes the vertices incident with the edges in $M$ while $G[V_M]$ denotes the graph
induced by those vertices.
A \emph{matching} $M$ of $G$ is a subset of its edges such that every vertex of $G$ is incident with at most one edge in $M$.

For $x, y \in \mathbb{Z}$, let $[x, y] = \setdef{z \in \mathbb{Z}}{x \leq z \leq y}$ while $[x] = [1,x]$.
Given a (partial) function $f \colon A \to B$, for every $b \in B$ we denote by $f^{-1}(b)$ the set of preimages of $b$ under $f$,
that is, $f^{-1}(b) = \setdef{a \in A}{f(a) = b}$.
Notice that in the case of partial functions it is not necessary for $f$ to be defined over the whole set $A$.
For a (partial) function $f$ we denote by $f[v \mapsto \alpha]$ the (partial) function
$(f \setminus \{(v, f(v))\}) \cup \{(v, \alpha)\}$, viewing $f$ as a set.
Standard $\sO$ notation is used to suppress polynomial factors.
Proofs of statements marked with {\appsymbNote} are deferred to the appendix.

\subparagraph{Cutwidth.}
A linear arrangement of a graph $G = (V,E)$ is an injective function $\pi \colon V \to [|V|]$,
that we also refer to as an \emph{ordering} of~$V$.
For an ordering $\pi$ and $i \in [0,n]$, we define $V_i = \setdef{v \in V}{\pi(v) \le i}$
and $\overline{V}_i = V \setminus V_i$;
notice that $V_0 = \varnothing$.
Let $E_i \subseteq E$ be the set of edges with one endpoint in $V_i$ and the other in $\overline{V}_i$,
for $i \in [0,n]$; we say that those edges \emph{cross} the cut.
In that case, the \emph{cutwidth} of $\pi$ is defined as $\ctw(\pi) = \max_{i \in [0,n]} |E_i|$.
The \emph{cutwidth} of $G$ is defined as $\ctw(G) = \min_{\pi} \ctw(\pi)$,
where the minimum is taken over all linear arrangements of~$G$.
It is well-known that for every graph $G$ it holds that $\pw(G) \leq \ctw(G)$.

\subparagraph{Clique-width.}
A graph of clique-width $k$ can be constructed through a sequence of the following operations on vertices that are labeled with at most $k$ different labels.
We can use (1) introducing a single vertex~$v$ of an arbitrary label~$i$, denoted $i(v)$,
(2) disjoint union of two labeled graphs, denoted $H_1 \oplus H_2$,
(3) introducing edges between \emph{all} pairs of vertices of two distinct labels~$i$ and~$j$ in a labeled graph~$H$, denoted $\eta_{i,j}(H)$,
and (4) changing the label of \emph{all} vertices of a given label~$i$ in a labeled graph~$H$ to a different label~$j$,
denoted $\rho_{i \to j}(H)$.
An expression describes a graph~$G$ if $G$ is the final graph given by the expression (after we remove all the labels).
The \emph{width} of an expression is the number of different labels it uses.
The clique-width of a graph is the minimum width of an expression describing it~\cite{dam/CourcelleO00}.
For a labeled graph $H$ and $v \in V(H)$, let $\lab_H(v)$ denote the label of $v$ in $H$,
while $\lab^{-1}_H (i) = \setdef{v \in V(H)}{\lab_H(v) = i}$ denotes the set of vertices of $H$ of label $i$.
For $S \subseteq V(H)$, let $H[S]$ denote the labeled subgraph of $H$ induced by $S$.
A clique-width expression is \emph{irredundant} if whenever the operation $\eta_{i,j}$
is applied on a graph $G$, there is no edge between an $i$-vertex and a $j$-vertex in $G$,
and we will use such expressions to simplify our algorithms.
We remark that any clique-width expression can be transformed in linear time into an irredundant one of the same width~\cite{dam/CourcelleO00}.

\subparagraph{pw-SETH.} The \emph{primal pathwidth SETH} (pw-SETH) states that,
for all $\varepsilon>0$, 3-SAT requires time at least
$(2-\varepsilon)^{\pw}n^{\bO(1)}$, where $\pw$ is the pathwidth of the primal
graph of the input formula.%
\footnote{The \emph{primal} (or \emph{Gaifman}) graph of a formula contains a vertex for each variable
and an edge when two variables appear in a common constraint.}
In other words, the pw-SETH posits that the simple
DP algorithm which solves 3-SAT in time $\sO(2^{\pw})$ is best possible.
Beating this algorithm seems to encapsulate the difficulty of improving upon
simple DP algorithms in general and indeed the pw-SETH is \emph{equivalent} to
many tight lower bounds for problems parameterized by linear structure widths
(pathwidth or linear clique-width) \cite{soda/Lampis25}. The pw-SETH also has
the advantage of being implied by several other standard assumptions, such as the
Set Cover Conjecture and the SETH for circuits of depth $\varepsilon n$~\cite{arxiv/Lampis24},
making lower bound results based on the pw-SETH seem more believable.

In this paper we present two such lower bound results based on the pw-SETH, in
one case (for \textsc{Induced Matching} parameterized by pathwidth) showing
that breaking our bound is \emph{equivalent} to the pw-SETH. For this, we
recall some notions and results from~\cite{soda/Lampis25} which we will make
use of later in our reductions, and in particular in
\cref{lem:induced:lb:CSP->Induced,thm:acyclic:lb}. Informally, reductions in
this context start from a constraint satisfaction problem \textsc{CSP} with
alphabet size equal to the desired base of the lower bound, similarly to
SETH-based reductions \cite{siamdm/Lampis20}. However, rather than reducing
from a \textsc{CSP} instance parameterized by the number of variables we reduce
from an instance parameterized by the pathwidth of its primal graph.
To facilitate the reduction, we use the fact that it is hard
(under pw-SETH) to distinguish between \textsc{CSP} instances which are
satisfiable and instances which are unsatisfiable even if for the majority of
variables we are allowed to select multiple assignments, albeit while only
modifying values in a monotone way along the given path decomposition.

\begin{definition}[{\cite[Definition~3.2]{soda/Lampis25}}]
    Suppose we have a \textsc{CSP} instance $\psi$ over an alphabet $ [B]$ of size $B \ge 2$,
    with variable set $V$, a path decomposition of its primal graph $B_1,\ldots, B_t$,
    and an injective function $b$ mapping each constraint to the index of a bag that contains all its variables.
    A \emph{multi-assignment} is a function $\sigma$ that takes as input a variable
    $x \in V$ and an index $j \in [t]$ such that $x \in B_j$ and returns a value in $[B]$.
    We will say that:
    \begin{enumerate}
        \item A multi-assignment $\sigma$ is \emph{satisfying} for $\psi$ if for each constraint $c$,
        the assignment $\sigma_c(x) = \sigma(x,b(c))$, that is,
        the restriction of the multi-assignment to $b(c)$,
        satisfies the constraint.

        \item A multi-assignment $\sigma$ is \emph{monotone} if for all $x \in V$ and $j_1<j_2$
        with $x \in B_{j_1} \cap B_{j_2}$ we have $\sigma(x,j_1) \le \sigma(x,j_2)$.

        \item A multi-assignment $\sigma$ is \emph{consistent} for $x \in V$ if for all
        $j_1,j_2 \in [t]$ such that $x \in B_{j_1} \cap B_{j_2}$
        we have $\sigma(x,j_1)=\sigma(x,j_2)$.
    \end{enumerate}
\end{definition}

\begin{corollary}[{\cite[Corollary 3.1]{soda/Lampis25}}]\label{cor:weird}
    For all $\varepsilon>0, B \ge 2$ we have the following.
    Suppose there is an algorithm with the following properties:
    \begin{itemize}
        \item it takes as input a $4$-\textsc{CSP} instance $\psi$ over
        an alphabet of size $B$, a partition of its variables into two sets $V_1, V_2$,
        a path decomposition of its primal graph of width $p$ where each bag contains at most
        $\bO (B \log p)$ variables of $V_2$, and an injective function $b$ mapping each constraint to a bag that contains its variables;

        \item it decides if there exists a monotone satisfying multi-assignment $\sigma$,
        which is consistent for the variables of $V_2$;

        \item it runs in time $\bO ( (B-\varepsilon)^p |\psi|^{\bO(1)})$.
    \end{itemize}
    Then the pw-SETH is false.
\end{corollary}

\section{Induced Matching}\label{sec:induced}

In this section we consider the {\InducedM} problem.
Chaudhary and Zehavi~\cite{siamdm/ChaudharyZ25} developed an algorithm of running time $\sO(3^\tw)$,
where $\tw$ denotes the treewidth of the input graph, and left as an open question whether this is optimal.
As our first result, in \cref{thm:induced:lb} we show that this is indeed the case,
as obtaining an algorithm of running time $\sO((3 - \varepsilon)^\pw)$, or even $\sO((3 - \varepsilon)^\ctw)$, is equivalent to falsifying the pw-SETH.
Next, we develop a DP algorithm of running time $\sO(3^\cw)$ for the problem, where $\cw$ denotes the clique-width of the input graph.
Since $\cw(G) \le \pw(G) + 2$ for any graph $G$, our algorithm is optimal under the pw-SETH due to \cref{thm:induced:lb}.

\subsection{Lower Bound}\label{subsec:induced:lb}

We first state the main theorem of this section,
which we then prove in \cref{lem:induced:lb:Induced->CSP,lem:induced:lb:CSP->Induced}.

\begin{theorem}\label{thm:induced:lb}
    The following are equivalent:
    \begin{itemize}
        \item there exists an algorithm deciding {\InducedM} in time $\sO((3-\varepsilon)^\pw)$,
        where $\pw$ denotes the pathwidth of the input graph, for some $\varepsilon > 0$,

        \item there exists an algorithm deciding {\InducedM} in time $\sO((3-\varepsilon)^\ctw)$,
        where $\ctw$ denotes the cutwidth of the input graph, for some $\varepsilon > 0$,

        \item the pw-SETH is false.
    \end{itemize}
\end{theorem}

\begin{lemmarep}[\appsymb]\label{lem:induced:lb:Induced->CSP}
    If the pw-SETH is false,
    then there exists $\varepsilon > 0$
    and an algorithm that solves {\InducedM} in time $\sO((3-\varepsilon)^\pw)$,
    where $\pw$ denotes the pathwidth of the input graph.
\end{lemmarep}

\begin{proof}
    Recall that in the \textsc{MaxW-CSP} problem we are given a
    \textsc{CSP} instance of arity $r$ and over an alphabet $\mathcal{B}$,
    as well as a weight function $w \colon \mathcal{B} \to \mathbb{N}$ and an integer $w_t$,
    and we want to decide whether there exists a satisfying assignment
    whose total weight is at least $w_t$.
    The weight of an assignment is defined as the sum of the weights
    of the assignments of its variables.

    Given an instance $(G,\ell)$ of {\InducedM} and a path decomposition of $G=(V,E)$ of width $\pw$,
    we will reduce it to an instance $\psi$ of \textsc{MaxW-CSP} of arity $r = \bO(1)$ and over an alphabet $\mathcal{B}$ of size $3$,
    where the pathwidth of the primal graph of $\psi$ is $\pw + \bO(1)$.
    If the pw-SETH is false, then by~\cite[Theorem~3.2]{soda/Lampis25} there exist $\varepsilon > 0$
    such that there is an algorithm that decides $\psi$ in time $\bO((|\mathcal{B}| - \varepsilon)^{\pw + \bO(1)} |\psi|^{\bO(1)})$,
    which allows us to obtain the desired running time for \InducedM.

    Suppose we have a nice path decomposition of $G$ of width $\pw$ with the bags numbered
    $B_1, \ldots, B_t$.
    Let the vertices of $G$ be numbered, that is, $V = \{v_1, \ldots, v_n\}$.
    We can assume that there exists an injective function $b \colon E \to [t]$
    which maps each edge of $G$ to the index of a bag that contains both of its endpoints
    ($b$ can be made injective by repeating bags if necessary).
    Furthermore, we can assume that $b$ maps edges to bags which are not introducing new vertices,
    again by repeating bags.
    We construct a CSP instance $\psi$ over an alphabet $\mathcal{B} = \{0,1,2\}$ as follows:
    \begin{enumerate}
        \item For each $v_i \in V$ we introduce $t$ variables $x_{i,j}$ for $j \in [t]$.

        \item For each $v_i \in V$ and $j \in [t-1]$ we add the following constraints:
        \begin{itemize}
            \item $x_{i,j} \neq 2 \implies x_{i,j+1} = x_{i,j}$,
            \item $x_{i,j} = 2 \implies x_{i,j+1} \neq 0$,
            \item If no edge $e$ has $b(e) = j$ or for some edge $e$ we have
                    $b(e) = j$ but $e$ is not incident on $v_i$,
                    then we also add the constraint $x_{i,j+1}=x_{i,j}$.
        \end{itemize}

        \item For each $v_i \in V$, if the last bag containing $v_i$ is $B_j$,
        then for all $j' \in [j,t-1]$ we add the constraint $x_{i,j'+1} = x_{i,j'}$.

        \item For each $e \in E$ let $j=b(e)$ and $e = \{v_{i_1}, v_{i_2}\}$ with $v_{i_1}, v_{i_2} \in B_j$.
        We add the constraints
        \begin{itemize}
            \item $(x_{i_1,j} = 0) \lor (x_{i_2,j} = 0) \lor (x_{i_1,j} = x_{i_2,j} = 2)$,
            \item $(x_{i_1,j} = 0) \lor (x_{i_2,j} = 0) \implies \big((x_{i_1,j+1} = x_{i_1,j}) \land (x_{i_2,j+1} = x_{i_2,j})\big)$,
            \item $(x_{i_1,j} = x_{i_2,j} = 2) \implies (x_{i_1,j+1} = x_{i_2,j+1} = 1)$.
        \end{itemize}

        \item For each $v_i \in V$ we add the constraints $x_{i,t} \neq 2$ and $x_{i,1} \neq 1$.
    \end{enumerate}
    Intuitively, we have a one-to-one mapping between assignments of the variables of $\psi$ and the induced matchings of $G$ so that
    (i) $x_{i,j} = 0$ if $v_i$ does not belong to the considered induced matching,
    (ii) $x_{i,j} = 1$ if $v_i$ belongs to the considered induced matching and its incident edge has been introduced by bag $B_j$,
    and (iii) $x_{i,j} = 2$ if $v_i$ belongs to the considered induced matching and its incident edge is introduced \emph{after} bag $B_j$.

    We define a weight function $w \colon \mathcal{B} \to \{0,1\}$ such that
    $w(1) = w(2) = 1$ and $w(0) = 0$.
    The target weight is $w_t = 2\ell \cdot t$,
    where $\ell$ is the target value for the given {\InducedM} instance.

    \begin{claim}
        If $G$ has an induced matching of size $\ell$,
        then $\psi$ has a satisfying assignment of weight at least $2\ell \cdot t$.
    \end{claim}

    \begin{claimproof}
        Let $M \subseteq E(G)$ be an induced matching of $G$ of size $\ell$.
        For each $v_i \notin V_M$ we assign to $x_{i,j}$ the value $0$, for all $j \in [t]$.
        As for the vertices of $V_M$, let $e \in M$ with $e = \{v_{i_1},v_{i_2}\}$ and $b(e)=j$.
        Then we assign to both $x_{i_1,j'},x_{i_2,j'}$ the value $2$ for all $j' \in [j]$,
        and the value $1$ for all $j' \in [j+1,t]$.

        It is easy to see that this assignment is of weight $|V_M| \cdot t \ge 2\ell \cdot t$.
        We argue that it is also satisfying.
        It is easy to see that the constraints introduced in Steps~2,~3, and~5 are satisfied.
        To see that the constraints introduced in Step~4 are also satisfied,
        let $e \in E$ with $e = \{v_{i_1},v_{i_2}\}$ and $b(e)=j$.
        If $e \in M$, then the corresponding constraints can be easily seen to be satisfied.
        Assume that $e \notin M$.
        Since $M$ is an induced matching, it holds that at most one endpoint of $e$ belongs to $V_M$,
        consequently it holds that $(x_{i_1,j} = 0) \lor (x_{i_2,j} = 0)$ and the constraints are once
        again satisfied.
    \end{claimproof}

    \begin{claim}
        If $\psi$ has a satisfying assignment of weight at least $2\ell \cdot t$,
        then $G$ has an induced matching of size $\ell$.
    \end{claim}

    \begin{claimproof}
        Let $i \in [n]$ such that $\psi$ assigns to $x_{i,t}$ the value $0$.
        In that case, due to the constraints in Step~$1$, it follows that $\psi$ assigns to
        $x_{i,j}$ the value $0$ for all $j \in [t]$.
        Let $I \subseteq [n]$ such that $x_{i,t}$ gets value $1$ by $\psi$;
        due to the discussion so far as well as the constraints introduced in Step 5, it follows that $|I| \ge 2\ell$.
        It suffices to show that $G' = G[\setdef{v_i}{i \in I}]$ is an $1$-regular subgraph of $G$.
        Let $i \in I$.
        First assume that there is no edge incident with $v_i$ in $G'$.
        In that case, for any edge $e = \{ v_i, v_{i'} \}$ in $G$ it holds that
        $i' \notin I$, that is, $x_{i',j}$ receives value $0$ for all $j \in [t]$.
        Consequently, by the constraints of Steps~2,~4, and~5, $x_{i,t}$ gets value $2$, a contradiction.
        Now assume that there are more than one edges incident with $v_i$ in $G'$,
        and let $e_1,e_2 \in E$ such that $e_1 = \{v_i, v_{i_1}\}$ and $e_2 = \{v_i, v_{i_2}\}$ with $i_1,i_2 \in I$,
        while $b(e_1) = j_1$ and $b(e_2) = j_2 > j_1$.
        In that case, due to the constraints in Steps~2,~3, and~4 it holds that
        $x_{i,j}$ gets value $1$ for all $j \in [j_1+1,t]$,
        in which case it follows that $x_{i,j_2} = 1$ and $x_{i_2,j_2} \neq 0$,
        falsifying the constraint $(x_{i,j_2} = 0) \lor (x_{i_2,j_2} = 0) \lor (x_{i,j_2} = x_{i_2,j_2} = 2)$,
        a contradiction.
    \end{claimproof}

    Finally, let us argue about the pathwidth of $\psi$.
    Take the path decomposition of $G$ and replace in each $B_j$ each vertex $v_i \in B_j$ with the variable $x_{i,j}$.
    For each $e \in E$ such that $e = \{v_{i_1}, v_{i_2}\}$ and $b(e)=j$ also place
    $x_{i_1,j+1}, x_{i_2,j+1}$ in $B_j$.
    This covers the constraints of Step~4,
    increasing the width by at most~2 (since $b$ is injective).
    For each $j \in [t-1]$ we insert between $B_j$ and $B_{j+1}$ a sequence of bags which start with
    $B_j$ and at each step add a variable $x_{i,j+1}$ and then remove $x_{i,j}$, one by one.
    Finally, to cover the remaining variables and constraints of Steps~2,~3, and~5,
    it suffices for each $v_i \in V$ that appears in the interval $B_{j_1}, \ldots, B_{j_2}$ to insert
    to the left of $B_{j_1}$ a sequence of bags that contain all of $B_{j_1}$ and a path decomposition of the path formed
    by $x_{i,1}, \ldots, x_{i,j_1}$, and similarly after $B_{j_2}$.
\end{proof}

For the other direction, we reuse the gadgets of the reduction presented by Lampis and Vasilakis~\cite{toct/LampisV24}
to prove that {\BDD} when $\Delta = 1$ cannot be solved in $\sO((3-\varepsilon)^{\pw})$ under the SETH;
the latter problem asks, given a graph $G$, $\Delta \ge 0$, and $k>0$,
whether there exists a set $S \subseteq V(G)$ of size at most $k$
such that the maximum degree of $G-S$ is at most $\Delta$.

\begin{lemmarep}[\appsymb]\label{lem:induced:lb:CSP->Induced}
    If there exists $\varepsilon>0$ and an algorithm that solves {\InducedM}
    in time $\sO((3-\varepsilon)^\ctw)$,
    where $\ctw$ denotes the cutwidth of the input graph,
    then the pw-SETH is false.
\end{lemmarep}

\begin{proof}
    We present a reduction from the $4$-CSP problem of \cref{cor:weird} to \InducedM.
    We are given a \textsc{CSP} instance $\psi$ whose variables take values from $[3]$,
    a partition of its variables into two sets $V_1, V_2$,
    a path decomposition $B_1, \ldots, B_t$ of $\psi$ of width $p$
    such that each bag contains $\bO(\log p)$ variables of $V_2$,
    and an injective function $b$ mapping each constraint to a bag that contains its variables.
    We want to construct an {\InducedM} instance $(G,\ell)$ with $\ctw(G) = p + o(p)$,
    such that if $\psi$ is satisfiable, then $G$ has an induced matching of size $\ell$,
    while if $G$ has an induced matching of size $\ell$,
    $\psi$ admits a monotone satisfying multi-assignment which is consistent for $V_2$.
    We assume that the variables of $\psi$ are numbered $X = \{ x_1,\ldots, x_n \}$,
    the constraints are $c_1, \ldots, c_m$,
    and that the given path decomposition is nice.
    We construct $G$ as follows.

    \proofsubparagraph{Block and Variable Gadgets.}
    For every variable $x_i$ and every bag with $x_i \in B_j$, where $j \in [j_1,j_2]$,
    construct a \emph{block gadget} $\hat{B}_{i,j}$.
    Here we consider two cases, depending on whether $x_i \in V_1$ or $x_i \in V_2$.

    If $x_i \in V_1$ then we construct the block gadget $\hat{B}_{i,j}$ as a path on vertices $p^{i,j}_1, p^{i,j}_2,p^{i,j}_3,p^{i,j}_4$.
    Furthermore, for all $j \in [j_1,j_2-1]$, we identify the vertices $p^{i,j}_4, p^{i,j+1}_1$ so that they are the same vertex.
    The resulting path is called the \emph{variable gadget} of $x_i$, and is denoted by $\hat{P}_i$.
    See also \cref{fig:induced:lb:block:V1}.

    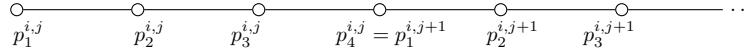
\begin{figure}[ht]
        \centering
        \begin{tikzpicture}[scale=0.8, transform shape]

        %%%%%%%%%% vertices and text
        \node[vertex] (vl11) at (0.5,5.5) {};
        \node[] () at (0.7,5.1) {$p^{i,j}_1$};

        \node[vertex] (vl12) at (2.5,5.5) {};
        \node[] () at (2.7,5.1) {$p^{i,j}_2$};

        \node[vertex] (vl13) at (4.5,5.5) {};
        \node[] () at (4.3,5.1) {$p^{i,j}_3$};

        \node[vertex] (vl14) at (6.5,5.5) {};
        \node[] () at (6.7,5.1) {$p^{i,j}_4 = p^{i,j+1}_1$};

        \node[vertex] (vl15) at (8.5,5.5) {};
        \node[] () at (8.7,5.1) {$p^{i,j+1}_2$};

        \node[vertex] (vl16) at (10.5,5.5) {};
        \node[] () at (10.3,5.1) {$p^{i,j+1}_3$};

        \node[] (vl17) at (12.5,5.5) {$\ldots$};

        %%%%%%%%% edges / arcs

        \draw[] (vl11)--(vl12)--(vl13)--(vl14)--(vl15)--(vl16)--(vl17);

        \end{tikzpicture}
        \caption{Part of the construction of the variable gadget $\hat{P}_i$, where $x_i \in V_1$.}
        \label{fig:induced:lb:block:V1}
    \end{figure}

    As for the case where $x_i \in V_2$, then the block gadget $\hat{B}_{i,j}$
    is a clique on vertices $p^{i,j},p^{i,j}_1,p^{i,j}_2,p^{i,j}_3$.
    Furthermore, for all $j \in [j_1,j_2-1]$ and $k \in [3]$,
    we add all edges between $p^{i,j}_k$ and $\{p^{i,j+1}_1,p^{i,j+1}_2,p^{i,j+1}_3\} \setminus \{p^{i,j+1}_k\}$,
    and we call the \emph{variable gadget} $\hat{P}_i$ this sequence of serially connected block gadgets.
    See also \cref{fig:induced:lb:block:V2}.

    \begin{figure}[ht]
        \centering
        \begin{tikzpicture}[scale=0.8, transform shape]

        %%%%%%%%%% vertices and text
        \node[vertex] (vl11) at (0.5,6.5) {};
        \node[] () at (0.6,6.8) {$p^{i,j}$};

        \node[vertex] (vl12) at (2.5,3.5) {};
        \node[] () at (2.6,3.1) {$p^{i,j}_1$};

        \node[vertex] (vl13) at (2.5,5.5) {};
        \node[] () at (2.2,5.8) {$p^{i,j}_2$};

        \node[vertex] (vl14) at (2.5,7.5) {};
        \node[] () at (2.6,7.9) {$p^{i,j}_3$};

        \node[vertex] (vl22) at (4.5,3.5) {};
        \node[] () at (4.7,3.1) {$p^{i,j+1}_1$};

        \node[vertex] (vl23) at (4.5,5.5) {};
        \node[] () at (5,5.8) {$p^{i,j+1}_2$};

        \node[vertex] (vl24) at (4.5,7.5) {};
        \node[] () at (4.7,7.9) {$p^{i,j+1}_3$};

        \node[vertex] (vl21) at (6.5,6.5) {};
        \node[] () at (6.6,6.8) {$p^{i,j+1}$};

        %%%%%%%%% edges / arcs

        \draw[] (vl11) edge [bend right] (vl12);
        \draw[] (vl11) edge [bend right] (vl13);
        \draw[] (vl11) edge [bend right] (vl14);
        \draw[] (vl12)--(vl13)--(vl14);
        \draw[] (vl12) edge [bend left] (vl14);

        \draw[] (vl21) edge [bend left] (vl22);
        \draw[] (vl21) edge [bend left] (vl23);
        \draw[] (vl21) edge [bend left] (vl24);
        \draw[] (vl22)--(vl23)--(vl24);
        \draw[] (vl22) edge [bend right] (vl24);

        \draw[] (vl12)--(vl23);
        \draw[] (vl12)--(vl24);

        \draw[] (vl13)--(vl22);
        \draw[] (vl13)--(vl24);

        \draw[] (vl14)--(vl22);
        \draw[] (vl14)--(vl23);

        \end{tikzpicture}
        \caption{Part of the construction of the variable gadget $\hat{P}_i$, where $x_i \in V_2$.}
        \label{fig:induced:lb:block:V2}
    \end{figure}
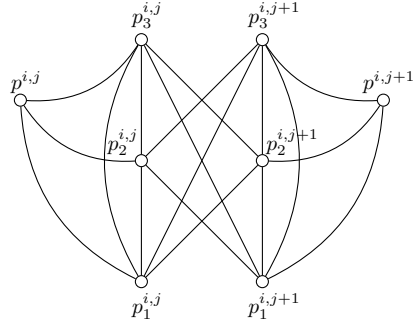

    \proofsubparagraph{Constraint Gadget.}
    This gadget is responsible for determining constraint satisfaction,
    based on the choices made in the rest of the graph.
    Let $c$ be a constraint of $\psi$, where $b(c) = j$ and
    $c$ involves variables $x_{i_1},x_{i_2},x_{i_3}, x_{i_4} \in B_j$.
    Consider the set $\mathcal{S}_c$ of the at most $3^4$ satisfying assignments of $c$.
    We construct the \emph{constraint gadget} $\hat{C}_c$ as follows.
    For each $\sigma \in \mathcal{S}_c$ construct a vertex $y_{c,\sigma}$.
    Additionally construct a vertex $w_c$,
    and add edges so that $w_c$ along with the vertices $y_{c,\sigma}$, for $\sigma \in \mathcal{S}_c$,
    form a clique.
    For each $\sigma \in \mathcal{S}_c$ and $\alpha \in [4]$,
    consider the following cases:
    \begin{itemize}
        \item if $x_{i_\alpha} \in V_1$ and $\sigma(x_{i_\alpha}) = 1$,
        then connect $y_{c,\sigma}$ to $p^{i,j}_3$,

        \item if $x_{i_\alpha} \in V_1$ and $\sigma(x_{i_\alpha}) = 2$,
        then connect $y_{c,\sigma}$ to $p^{i,j}_1$ and $p^{i,j}_4$,

        \item if $x_{i_\alpha} \in V_1$ and $\sigma(x_{i_\alpha}) = 3$,
        then connect $y_{c,\sigma}$ to $p^{i,j}_2$,

        \item if $x_{i_\alpha} \in V_2$ and $\sigma(x_{i_\alpha}) = k \in [3]$,
        then connect $y_{c,\sigma}$ to vertices $\{p^{i,j}_1,p^{i,j}_2,p^{i,j}_3\} \setminus \{p^{i,j}_k\}$.
    \end{itemize}

    This completes the construction of $G$.
    We set $\ell = L_1+L_2+m$, where $L_i$ denotes the number of block gadgets constructed
    due to variables belonging to $V_i$, and $m$ is the number of constraints of $\psi$.
    In that case, $(G,\ell)$ is the constructed instance of \InducedM.

    \begin{claim}\label{claim:induced:lb:CSP->Induced:direction1}
        If $\psi$ is satisfiable, then $G$ has an induced matching of size $\ell$.
    \end{claim}

    \begin{claimproof}
        Let $\sigma$ be a satisfying assignment for $\psi$.
        For each $x_i \in V_2$ we select into our solution $M$ the edge $\{p^{i,j}, p^{i,j}_{\sigma(x_i)}\}$,
        for all $j \in [t]$ with $x_i \in B_j$.
        For each $x_i \in V_1$ and $j \in [t]$ such that $x_i \in B_j$
        we select into $M$ the edge $\{p^{i,j}_{\sigma(x_i)}, p^{i,j}_{\sigma(x_i)+1}\}$.
        Finally, for each constraint $c$,
        we select into $M$ the edge $\{y_{c,\sigma_c}, w_c\}$,
        where $\sigma_c \in \mathcal{S}_c$ is the restriction of $\sigma$ to the variables of $c$.

        $M$ is a matching of size $\ell$.
        It remains to argue that it is an induced matching.
        To this end, observe that among the vertices of $V_M$,
        none of the vertices belonging to a constraint gadget
        is neighbors with vertices of $V_M$ that belong to a block gadget.
        The statement then easily follows.
    \end{claimproof}

    \begin{claim}\label{claim:induced:lb:CSP->Induced:direction2}
        If $G$ has an induced matching of size $\ell$,
        then $\psi$ admits a monotone satisfying multi-assignment which is consistent for $V_2$.
    \end{claim}

    \begin{claimproof}
        Let $M$ be an induced matching of $G$ of size at least $\ell$.
        Let $c$ be a constraint of $\psi$.
        We assume without loss of generality that either $M$ does not contain any edge incident with vertices of $\hat{C}_c$,
        or there exists a $\sigma_c \in \mathcal{S}_c$ such that $\{y_{c,\sigma_c}, w_c\} \in M$.
        Indeed, assume there exists $e \in M$ such that $y_{c,\sigma_0} \in e$ for some $\sigma_0 \in \mathcal{S}_c$,
        while $w_c \notin e$.
        Then it holds that $V_M$ contains no other vertex of $\hat{C}_c$,
        as its vertices form a clique.
        Consequently, $(M \setminus \{e\}) \cup \{\{y_{c,\sigma_0}, w_c \}\}$ is also an induced matching of the same size.

        Now consider a variable gadget $\hat{P}_i$.
        First consider the case where $x_i \in V_1$.
        It is easy to see that $M$ contains at most $1$ edge per block gadget,
        with all edges having both of their endpoints in the same block gadget.
        Now assume that $x_i \in V_2$ and $x_i \in B_j$ for all $j \in [j_1,j_2]$.
        Assume that there exists an edge $e \in M$ such that $e = \{p^{i,j}_k,p^{i,j+1}_{k'}\}$,
        with $j \in [j_1,j_2-1]$ and $k,k' \in [3]$.
        Then, it follows that no other vertex of $\hat{B}_{i,j}$ and $\hat{B}_{i,j+1}$ belongs to $V_M$,
        and exchanging $\{p^{i,j}_k,p^{i,j+1}_{k'}\}$ with edges $\{p^{i,j}_k,p^{i,j}\}$ and $\{p^{i,j+1}_{k'},p^{i,j+1}\}$
        results in an induced matching of larger size.
        In the following assume that for $M$ it holds that any edge incident with vertices of $\hat{B}_{i,j}$ has
        both of its endpoints in said gadget. Notice that $M$ contains at most one edge from each such block gadget.

        Consequently, it follows that $M$ contains at most $L_1 + L_2 + m = \ell$ edges,
        where $L_i$ denotes the number of block gadgets constructed due to variables belonging
        to $V_i$, and $m$ is the number of constraints of $\psi$.
        It follows that $|M| = \ell$, and $M$ contains exactly one edge per block gadget
        as well as exactly one edge per constraint gadget.

        Consider a multi-assignment $\sigma \colon X \times [t] \to [3]$ over all pairs $x_i \in X$ and $j \in [t]$
        with $x_i \in B_j$.
        In particular, consider two cases.
        If $x_i \in V_1$ and $\{p^{i,j}_k, p^{i,j}_{k+1}\} \in M$,
        then $\sigma(x_i,j) = k \in [3]$.
        For the case where $x_i \in V_2$, if $\{p^{i,j}_k, p^{i,j}\} \in M$,
        then $\sigma(x_i,j) = k \in [3]$.

        We argue that $\sigma$ is consistent for all $x_i \in V_2$.
        Fix $x_i \in V_2$ and $j \in [j_1,j_2-1]$, where $B_{j_1}$ and $B_{j_2}$ denote the first and last
        bag that contain $x_i$.
        Notice that $\{p^{i,j}_k, p^{i,j}\} \in M$ implies that $p^{i,j+1}_{k'} \notin V_M$
        where $k \in [3]$ and $k' \in [3] \setminus \{k\}$, thus $\sigma(x_i,j) = \sigma(x_i,j')$
        for all $j,j' \in [j_1,j_2]$.

        We next argue that $\sigma$ is monotone for $x_i \in V_1$,
        where $x_i \in B_j$ for all $j \in [j_1,j_2]$.
        Let $j \in [j_1,j_2-1]$.
        If $\{p^{i,j}_3, p^{i,j}_4\} \in M$,
        then it follows that $p^{i,j+1}_2 \notin V_M$,
        implying that $\{p^{i,j+1}_3, p^{i,j+1}_4\} \in M$.
        If on the other hand $\{p^{i,j}_2, p^{i,j}_3\} \in M$,
        then it follows that $p^{i,j+1}_1 \notin V_M$,
        implying that either $\{p^{i,j+1}_2, p^{i,j+1}_3\} \in M$ or $\{p^{i,j+1}_3, p^{i,j+1}_4\} \in M$.

        Finally, we argue that $\sigma$ is satisfying.
        Consider a constraint $c$ with $b(c)=j$,
        involving four variables from $V_{i_1}, V_{i_2}, V_{i_3}, V_{i_4}$.
        Let $\sigma_c \in \mathcal{S}_c$ such that $\{y_{c,\sigma_c},w_c\} \in M$.
        We claim that $\sigma_c$ must be consistent with our assignment.
        Indeed, if $\sigma_c$ assigns value $k \in [3]$ to $x_i \in V_2$,
        then it follows that $(\{p^{i,j}_1,p^{i,j}_2,p^{i,j}_3\} \setminus \{p^{i,j}_k\}) \notin V_M$,
        thus $\{p^{i,j}_k, p^{i,j}\} \in M$.
        On the other hand, if $\sigma_c$ assigns value $k \in [3]$ to $x_i \in V_1$,
        then
        \begin{itemize}
            \item if $k=1$, then $p^{i,j}_3 \notin V_M$, implying that $\{p^{i,j}_1,p^{i,j}_2\} \in M$,
            \item if $k=2$, then $p^{i,j}_1,p^{i,j}_4 \notin V_M$, implying that $\{p^{i,j}_2,p^{i,j}_3\} \in M$,
            \item if $k=3$, then $p^{i,j}_2 \notin V_M$, implying that $\{p^{i,j}_3,p^{i,j}_4\} \in M$.
        \end{itemize}
        This completes the proof.
    \end{claimproof}

    \begin{claim}\label{claim:induced:lb:CSP->Induced:width}
        It holds that the cutwidth of $G$ is at most $p + \bO(\log p)$.
    \end{claim}

    \begin{claimproof}
        Recall that $B_1, \ldots, B_t$ denotes the path decomposition of the primal graph of $\psi$.
        We construct a linear arrangement $\pi$ of $V(G)$ by starting with an empty linear arrangement
        and then proceed as follows by always adding the vertices to the very right.
        First we iterate over $j = 1, \ldots, t$.
        For a fixed value of $j$, we next iterate over all $x_i \in B_j$.
        For a fixed $x_i \in B_j$, we place all unplaced%
        \footnote{Note that for $j \ge 2$ and $x_i \in B_j \cap V_1$, the vertex $p^{i,j}_1$ is already placed into the linear arrangement
        as it coincides with $p^{i,j-1}_4$.}
        vertices of $\hat{B}_{i,j}$:
        if $x_i \in V_1$, then we place them consecutively, that is, like $p^{i,j}_1, p^{i,j}_2,p^{i,j}_3,p^{i,j}_4$;
        otherwise, i.e., if $x_i \in V_2$, we place them in an arbitrary order.
        After processing all $x_i \in B_j$, if there exists a constraint $c$ with $b(c) = j$,
        then we add all vertices of $\hat{C}_c$ in an arbitrary order.
        This completes the description of $\pi$.

        We now argue that the cutwidth of $\pi$ is $p + \bO(\log p)$.
        Consider an arbitrary cut. We proceed to bound the number of edges crossing this cut:
        \begin{itemize}
            \item Since every bag of the path decomposition of $\psi$ has at most $p$ variables of $V_1$,
            there are at most $p$ edges whose endpoints both belong to a variable gadget due to a variable of $V_1$.

            \item Since every bag of the path decomposition of $\psi$ has $\bO(\log p)$ variables of $V_2$,
            there are $\bO(\log p)$ edges whose endpoints both belong to a variable gadget due to a variable of $V_2$.

            \item Since every constraint gadget contains at most $(3^4)^2$ edges and no two constraint gadgets are connected,
            there are at most $3^8$ edges whose endpoints both belong to a constraint gadget.

            \item Since there are at most $3^4 \cdot 8$ edges connecting a constraint gadget with the block gadgets,
            while any vertex of a block gadget is connected with at most $2$ constraint gadgets,
            there are at most $3^4 \cdot 16$ edges whose one endpoint belongs a block gadget and the other to a constraint gadget.
        \end{itemize}
        Summing over the number of those edges results in the stated upper bound.
        %
        % PATHWIDTH BOUND
        % Finally, to bound the pathwidth of the graph $G$,
        % start with the decomposition of the primal graph of $\psi$ and
        % in each $B_j$ replace each $x_i \in V_2 \cap B_j$ with all the vertices of $\hat{B}_{i,j}$.
        % We further add in $B_j$ all the vertices of $\hat{B}_{i,j+1}$ for $x_i \in V_2 \cap B_{j+1}$.
        % For each constraint $c$, let $b(c)=j$ and add into $B_j$ all the at most $3^4 + 1$
        % vertices of $\hat{C}_c$.
        % So far each bag contains $p + \bO(\log p)$ vertices.
        % To cover the remaining vertices,
        % for each $j \in [t]$ we replace the bag $B_j$ with a sequence of bags such that
        % all of them contain the vertices we have added to $B_j$ so far,
        % the first bag contains $\setdef{p^{i,j}_1}{x_i \in V_1 \cap B_j}$ and
        % the last bag contains $\setdef{p^{i,j}_4}{x_i \in V_1 \cap B_j}$.
        % We insert a sequence of $\bO(p)$ bags between these two,
        % at each step adding a vertex $p^{i,j}_{k+1}$ and then removing $p^{i,j}_k$
        % in a way that covers all edges of paths due to block gadgets.
    \end{claimproof}
    \cref{lem:induced:lb:CSP->Induced} now follows by \cref{claim:induced:lb:CSP->Induced:direction1,claim:induced:lb:CSP->Induced:direction2,claim:induced:lb:CSP->Induced:width}.
\end{proof}

\subsection{Algorithm}\label{subsec:induced:cw}

In this section we prove the following theorem.

\begin{theorem}\label{thm:induced:cw}
    There is an algorithm that,
    given a graph $G$ and an irredundant clique-width expression $\psi$ of $G$ of width $\cw$,
    counts the number of induced matchings of $G$ of each size $\ell \in [0,n/2]$ in time $\sO(3^\cw)$.
\end{theorem}

\begin{proof}
    Before describing our algorithm we start with some definitions and notation.
    Given a tuple $\sigma$, we refer to its $w$-th entry by $\sigma(w)$.
    Let $H$ be a graph and $S \subseteq V(H)$ a subset of its vertices.
    We say that a vertex $v \in S$ is \emph{unsaturated} in $H[S]$ if its degree is $\deg_{H[S]} (v) = 0$,
    else if $\deg_{H[S]} (v) = 1$ it is \emph{saturated} in $H[S]$ instead.
    We say that $S$ is a \emph{partial matching} of $H$ if the maximum degree of $H[S]$ is at most $1$,
    in which case all of its vertices are either saturated or unsaturated in $H[S]$.
    The \emph{size of a partial matching} $S$ of $H$ is equal to the number of its vertices, that is, $|S|$.
    Finally, we say that a partial matching $S$ of $H$ is \emph{extensible} if there exists an induced matching $M$ of $G$
    such that $V_M \cap V(H) = S$.

    Let $\mathcal{H}$ denote the set of all $\cw$-labeled graphs generated by subexpressions of $\psi$.
    For $H \in \mathcal{H}$ and $i \in [\cw]$,
    we define the partial function $\sgn_{H,i} \colon 2^{V(H)} \to \{0,1,2\}$ for $S \subseteq V(H)$ as follows:
    \[
        \sgn_{H,i}(S) =
            \begin{cases}
                0   &\text{if $S \cap \lab^{-1}_H (i) = \varnothing$,}\\
                1   &\text{if $S \cap \lab^{-1}_H (i) \neq \varnothing$ and $\deg_{H[S]}(v) = 1$ for all $v \in S \cap \lab^{-1}_H (i)$,}\\
                2   &\text{if $S \cap \lab^{-1}_H (i) = \{v\}$ and $\deg_{H[S]}(v) = 0$.}
            \end{cases}
    \]
    We further define the tuple $\sgn_H(S) = \left( \sgn_{H,1}(S), \ldots, \sgn_{H,\cw}(S) \right)$ to be
    the \emph{signature} of $S$ in $H$ if $\sgn_{H,i}(S)$ is well-defined for all $i \in [\cw]$.
    Notice that $\sgn_H(S)$ is only defined for a subset of the partial matchings of $H$,
    while there are at most $3^{\cw}$ different signatures.
    For improved readability, we sometimes refer to tuples $\sigma \in \{0,1,2\}^{\cw}$ as signatures.
    We first argue that for any induced matching $M$ of $G$,
    the signature of $V_M \cap V(H)$ is well-defined for all $H \in \mathcal{H}$.

    \begin{claim}\label{claim:induced:cw:well_defined}
        Let $M$ be an induced matching of $G$.
        For all $H \in \mathcal{H}$ it holds that $\sgn_H (V_M \cap V(H))$ is well-defined.
    \end{claim}

    \begin{claimproof} %
        Notice that for all $H \in \mathcal{H}$ it holds that $H[V_M \cap V(H)]$ is a subgraph of $G[V_M]$,
        therefore the maximum degree of $H[V_M \cap V(H)]$ is at most $1$.
        Consequently, $V_M \cap V(H)$ is a partial matching of $H$.
        Fix $H \in \mathcal{H}$ and let $S = V_M \cap V(H)$.
        It suffices to prove that $\sgn_{H,i}(S)$ is well-defined for all $i \in [\cw]$.

        For the sake of contradiction, assume that $\sgn_{H,i}(S)$ is not well-defined for some $i \in [\cw]$.
        Since $S$ is a partial matching of $H$, that can only be the case when $S \cap \lab^{-1}_{H}(i)$ contains at least two vertices,
        say $v,v'$, and at least one of them, say $v$, is unsaturated in $H[S]$,
        that is, $\deg_{H[S]}(v) = 0$.
        Let $u \in V_M$ such that $\{v,u\} \in M$.
        Since $\deg_{H[S]}(v) = 0$ it holds that either $u \notin V(H)$ or $\{v,u\} \notin E(H)$.
        This means that at some point further in the clique-expression there exists a join operation with
        $H_1,H_2 \in \mathcal{H}$ and $H_2 = \eta_{i',j'}(H_1)$ that adds all edges between vertices of labels $i'$ and $j'$,
        with $v,v'$ having label $i'$ and $u$ having label $j'$.
        In that case however, it holds that the degree of $u \in V_M \cap V(H_2)$ in $H_2[V_M \cap V(H_2)]$
        is at least $2$ as it is adjacent to both $v$ and $v'$, a contradiction.
    \end{claimproof}

    Our algorithm proceeds by dynamic programming along $\psi$ in a bottom-up fashion.
    In particular, for every $H \in \mathcal{H}$ it stores a table $\DPt_H[\cdot,\cdot]$ where, on a high-level,
    for $k \in [0,n]$ and $\sigma \in \{0,1,2\}^{\cw}$, $\DPt_H[k,\sigma]$ contains
    the number of partial matchings $S$ of $H$ of size $k$ and signature $\sigma$,
    that is, with $|S|=k$ and $\sgn_H(S) = \sigma$.
    We now proceed to describe how to populate the DP tables, and in \Cref{lemma:induced:cw:join-correctness,lemma:induced:cw:relabel-correctness,lemma:induced:cw:union-correctness} we prove that they indeed contain the number of partial matchings of appropriate size and signature.

    \proofsubparagraph{Singleton $H = i(v)$.}
    For $H = i(v)$, notice that $\lab^{-1}_H(i) = \{v\}$ and $\lab^{-1}_H(w) = \varnothing$ for all $w \in [\cw] \setminus \{i\}$.
    Consequently, for $k \in [0,n]$ and $\sigma \in \{0,1,2\}^{\cw}$ we set
    \begin{equation}\label{eq:DPH-singleton}
        \DPt_H[k,\sigma] =
            \begin{cases}
                1   &\text{if $k=0$ and $\sigma(w) = 0$ for all $w \in [\cw]$,}\\
                1   &\text{if $k=1$ and $\sigma(w) = 0$ for all $w \in [\cw] \setminus \{i\}$ and $\sigma(i) = 2$,}\\
                0   &\text{otherwise}.
            \end{cases}
    \end{equation}
    Notice that all but two entries in $\DPt_H[\cdot,\cdot]$ are $0$,
    and by \cref{eq:DPH-singleton} we can fill the table $\DPt_H[\cdot,\cdot]$ in time $\sO(3^{\cw})$.

    \proofsubparagraph{Joining labels with edges, $H = \eta_{i,j}(H')$.}
    For each pair of distinct $i,j \in [\cw]$,
    we define a partial function $f_{i,j} \colon \{0,1,2\}^{\cw} \to \{0,1,2\}^{\cw}$.
    In particular, for $\sigma' \in \{0,1,2\}^{\cw}$ we set $f_{i,j}(\sigma') = \sigma$
    such that $\sigma(w) = \sigma'(w)$ for all $w \in [\cw] \setminus \{i,j\}$
    and
    \[
        (\sigma(i),\sigma(j)) =
            \begin{cases}
                (\sigma'(i), \, \sigma'(j)) &\text{if $0 \in \{\sigma'(i), \, \sigma'(j)\}$,}\\
                (1,1)                       &\text{if $\sigma'(i) = \sigma'(j) = 2$}.
            \end{cases}
    \]
    For a fixed size $k \in [0,n]$ and signature $\sigma \in \{0,1,2\}^{\cw}$ we set the value of $\DPt_H[k,\sigma]$ to be
    \begin{equation}\label{eq:DPH-join}
        \DPt_H[k,\sigma] = \sum_{\sigma' \in f^{-1}_{i,j} (\sigma)} \DPt_{H'}[k,\sigma'].
    \end{equation}
    If $f^{-1}_{i,j}(\sigma) = \varnothing$ we set $\DPt_H[k,\sigma] = 0$.
    Notice that computing $f^{-1}_{i,j}(\sigma)$ requires $\bO(\cw)$ time,
    thus by \cref{eq:DPH-join} we can fill the table $\DPt_H[\cdot,\cdot]$ in time $\sO(3^{\cw})$.

    \begin{lemmarep}[\appsymb]\label{lemma:induced:cw:join-correctness}
        Let $H = \eta_{i,j}(H')$.
        Assume that for all $k \in [0,n]$ and $\sigma \in \{0,1,2\}^{\cw}$,
        $\DPt_{H'}[k,\sigma]$ is equal to the number of partial matchings $S$ of $H'$ such that $|S|=k$ and $\sgn_{H'}(S) = \sigma$.
        Then, $\DPt_H[k,\sigma]$ is equal to the number of partial matchings $S$ of $H$ such that $|S|=k$ and $\sgn_H(S) = \sigma$.
    \end{lemmarep}

    \begin{proof}
        We show that the partial function $f_{i,j}$ describes precisely the effect of the operation $\eta_{i,j}$
        on the signature of a partial matching $S$.

        \begin{claim}\label{claim:induced:cw:join-correctness}
            Let $S \subseteq V(H')$ such that $\sgn_{H'}(S) = \sigma'$.
            Then, $\sgn_{H}(S) = \sigma$ if and only if $f_{i,j}(\sigma') = \sigma$.
        \end{claim}

        \begin{claimproof}
            First, notice that the labels of the vertices remain the same between $H'$ and $H$,
            therefore it holds that $\lab_H(v) = \lab_{H'}(v)$ for all $v \in V(H')$.
            This implies that $S \cap \lab^{-1}_H (w) = S \cap \lab^{-1}_{H'} (w)$ for all $w \in [\cw]$.
            Since $H[S]$ is obtained from $H'[S]$ by adding all edges between vertices of label~$i$ and~$j$,
            it follows that $\deg_{H[S]}(v) = \deg_{H'[S]}(v)$ for all $v \in S \setminus (\lab^{-1}_H(i) \cup \lab^{-1}_H(j))$,
            thus $\sgn_{H,w}(S) = \sgn_{H',w}(S)$ for all $w \in [\cw] \setminus \{i,j\}$.
            In that case, $\sgn_{H}(S)$ and $f_{i,j}(\sigma')$ agree on all labels $w \in [\cw] \setminus \{i,j\}$.
            It remains to argue for labels~$i$ and~$j$.

            Assume first that $0 \in \{\sigma'(i), \, \sigma'(j)\}$.
            It is easy to see that in that case the graphs $H[S]$ and $H'[S]$ are the same,
            consequently $\sgn_{H}(S) = \sgn_{H'}(S)$ follows, and by the definition of partial function $f_{i,j}$
            we have that $\sgn_{H}(S) = f_{i,j}(\sigma')$.

            Next, assume that $\sigma'(i) = \sigma'(j) = 2$.
            Then, it holds that both $S \cap \lab^{-1}_{H}(i)$ and $S \cap \lab^{-1}_{H}(i)$ are comprised of a single vertex
            that is of degree $0$ in $H'[S]$, which implies that $\sgn_{H,i}(S) = \sgn_{H,j}(S) = 1$
            since the same vertices have degree~$1$ in $H[S]$.
            Again we have that $\sgn_{H}(S) = f_{i,j}(\sigma')$.

            Lastly, assume that $0 \notin \{\sigma'(i), \, \sigma'(j)\}$ and either $\sigma'(i) \neq 2$ or $\sigma'(j) \neq 2$.
            Assume without loss of generality that $\sigma'(i) = 1$ and $\sigma'(j) \in \{1,2\}$,
            and let $v \in S \cap \lab^{-1}_{H'}(i)$ and $u \in S \cap \lab^{-1}_{H'}(j)$,
            where $\deg_{H'[S]}(v) = 1$.
            Since $\psi$ is by definition an irredundant clique-width expression,
            it follows that $\{v,u\} \notin E(H')$.
            Consequently, the operation $\eta_{i,j}$ adds the edge $\{v,u\}$ in $H$, thus $\deg_{H[S]}(v) > \deg_{H'[S]}(v) = 1$.
            In that case $\sgn_{H,i}(S)$ is undefined, as is the case for $f_{i,j}(\sigma')$.
        \end{claimproof}

        Consequently, the number of partial matchings of $H$ of size $k$ and signature $\sigma$
        is exactly equal to the sum of the number of partial matchings of $H'$ of size $k$ and signature $\sigma'$
        over all $\sigma' \in f^{-1}_{i,j}(\sigma)$.
        The statement now follows by the hypothesis for $\DPt_{H'}[\cdot,\cdot]$ and \cref{eq:DPH-join}.
    \end{proof}

    \proofsubparagraph{Relabeling, $H = \rho_{i \to j}(H')$.}
    For each pair of distinct $i,j \in [\cw]$,
    we define a partial function $g_{i \to j} \colon \{0,1,2\}^{\cw} \to \{0,1,2\}^{\cw}$.
    In particular, for $\sigma' \in \{0,1,2\}^{\cw}$ we set $g_{i \to j}(\sigma') = \sigma$
    such that $\sigma(w) = \sigma'(w)$ for all $w \in [\cw] \setminus \{i,j\}$,
    $\sigma(i) = 0$, and
    \[
        \sigma(j) =
            \begin{cases}
                \max \{ \sigma'(i), \, \sigma'(j) \}    &\text{if $0 \in \{\sigma'(i), \, \sigma'(j)\}$,}\\
                1                                       &\text{if $\sigma'(i) = \sigma'(j) = 1$}.
            \end{cases}
    \]
    For a fixed size $k \in [0,n]$ and signature $\sigma \in \{0,1,2\}^{\cw}$ we set the value of $\DPt_H[k,\sigma]$ to be
    \begin{equation}\label{eq:DPH-relabel}
        \DPt_H[k,\sigma] = \sum_{\sigma' \in g^{-1}_{i \to j} (\sigma)} \DPt_{H'}[k,\sigma'].
    \end{equation}
    If $g^{-1}_{i \to j}(\sigma) = \varnothing$ we set $\DPt_H[k,\sigma] = 0$.
    Notice that computing $g^{-1}_{i \to j}(\sigma)$ requires $\bO(\cw)$ time,
    thus by \cref{eq:DPH-relabel} we can fill the table $\DPt_H[\cdot,\cdot]$ in time $\sO(3^{\cw})$.

    \begin{lemmarep}[\appsymb]\label{lemma:induced:cw:relabel-correctness}
        Let $H = \rho_{i \to j}(H')$.
        Assume that for all $k \in [0,n]$ and $\sigma \in \{0,1,2\}^{\cw}$,
        $\DPt_{H'}[k,\sigma]$ is equal to the number of partial matchings $S$ of $H'$ such that $|S|=k$ and $\sgn_{H'}(S) = \sigma$.
        Then, $\DPt_H[k,\sigma]$ is equal to the number of partial matchings $S$ of $H$ such that $|S|=k$ and $\sgn_H(S) = \sigma$.
    \end{lemmarep}

    \begin{proof}

        We show that the partial function $g_{i \to j}$ describes precisely the effect of the operation $\rho_{i \to j}$
        on the signature of a partial matching $S$.

        \begin{claim}\label{claim:induced:cw:relabel-correctness}
            Let $S \subseteq V(H')$ such that $\sgn_{H'}(S) = \sigma'$.
            Then, $\sgn_{H}(S) = \sigma$ if and only if $g_{i \to j}(\sigma') = \sigma$.
        \end{claim}

        \begin{claimproof}
            First, notice that the labels of all vertices with label different than~$i$ remain the same between $H'$ and $H$,
            therefore it holds that $\lab_H(v) = \lab_{H'}(v)$ for all $v \in V(H') \setminus \lab^{-1}_{H'}(i)$.
            This implies that $S \cap \lab^{-1}_H (w) = S \cap \lab^{-1}_{H'} (w)$ for all $w \in [\cw] \setminus \{i,j\}$.
            Furthermore, notice that no new edge is added in $H$, consequently it holds that $\deg_{H[S]}(v) = \deg_{H'[S]}(v)$ for all $v \in V(H)$.
            This implies that $\sgn_{H,w}(S) = \sgn_{H',w}(S)$ for all $w \in [\cw] \setminus \{i,j\}$.
            Moreover, since $H$ is obtained from $H'$ by relabeling all vertices of label~$i$ to label~$j$,
            it holds that $\lab^{-1}_H(j) = \lab^{-1}_{H'}(i) \cup \lab^{-1}_{H'}(j)$ and $\lab^{-1}_H(i) = \varnothing$,
            thus $\sgn_{H,i}(S) = 0$.
            In that case, $\sgn_{H}(S)$ and $f_{i,j}(\sigma')$ agree on all labels $w \in [\cw] \setminus \{j\}$.
            It remains to argue for label~$j$.

            Assume first that $0 \in \{\sigma'(i), \, \sigma'(j)\}$.
            In that case, it holds that either $S \cap \lab^{-1}_{H'}(i) = \varnothing$ or $S \cap \lab^{-1}_{H'}(j) = \varnothing$.
            Assume that $S \cap \lab^{-1}_{H'}(i) = \varnothing$, that is, $\sgn_{H',i}(S) = 0$,
            which implies that $S \cap \lab^{-1}_{H}(j) = S \cap \lab^{-1}_{H'}(j)$,
            thus $\sgn_{H,j}(S) = \sgn_{H',j}(S) = \max \{\sgn_{H',j}(S), \, \sgn_{H',i}(S)\}$.
            In a similar way, when $S \cap \lab^{-1}_{H'}(i) = \varnothing$ we get that $\sgn_{H',j}(S) = 0$
            and $\sgn_{H,j}(S) = \sgn_{H',i}(S) = \max \{\sgn_{H',i}(S), \, \sgn_{H',j}(S)\}$.
            Thus, in this case we have that $\sgn_{H}(S) = f_{i,j}(\sigma')$.

            Next, assume that $\sigma'(i) = \sigma'(j) = 1$.
            Then, it holds that for all $v_1 \in S \cap \lab^{-1}_{H'}(i)$ and $v_2 \in S \cap \lab^{-1}_{H'}(j)$
            we have that $\deg_{H'[S]} (v_1) = \deg_{H'[S]} (v_2) = 1$.
            Since $\lab^{-1}_H(j) = \lab^{-1}_{H'}(i) \cup \lab^{-1}_{H'}(j)$ and
            for all $v \in V(H)$ it holds that $\deg_{H[S]}(v) = \deg_{H'[S]}(v)$,
            it follows that $\sgn_{H,j}(S) = 1$.
            Again we have that $\sgn_{H}(S) = f_{i,j}(\sigma')$.

            Lastly, assume that $0 \notin \{\sigma'(i), \, \sigma'(j)\}$ and either $\sigma'(i) \neq 1$ or $\sigma'(j) \neq 1$.
            Assume without loss of generality that $\sigma'(i) = 2$ and $\sigma'(j) \in \{1,2\}$.
            In that case, it holds that $S \cap \lab^{-1}_{H'}(i) = \{v\}$ with $\deg_{H'[S]}(v) = 0$
            as well as $S \cap \lab^{-1}_{H'}(i) \neq \varnothing$.
            Consequently it holds that $S \cap \lab^{-1}_H(j)$ contains at least $2$ vertices $v_1,v_2$, with $\deg_{H[S]}(v_1) = 0$,
            thus $\sgn_{H,j}(S)$ is undefined, as is the case for $f_{i,j}(\sigma')$.
        \end{claimproof}

        Consequently, the number of partial matchings of $H$ of size $k$ and signature $\sigma$
        is exactly equal to the sum of the number of partial matchings of $H'$ of size $k$ and signature $\sigma'$
        over all $\sigma' \in g^{-1}_{i \to j}(\sigma)$.
        The statement now follows by the hypothesis for $\DPt_{H'}[\cdot,\cdot]$ and \cref{eq:DPH-relabel}.
    \end{proof}

    \proofsubparagraph{Disjoint union, $H = H_1 \oplus H_2$.}
    For the Union nodes, an approach similar to the previous cases would yield a total running time of $\sO(9^{\cw})$
    to populate the table $\DPt_H[\cdot,\cdot]$, which can be further improved to $\sO(6^{\cw})$ by some slightly more
    careful analysis.
    In order to bring this down to $\sO(3^{\cw})$ we proceed in a different way
    by using the techniques introduced in~\cite{mfcs/BodlaenderLRV10,tcs/CyganP10,esa/RooijBR09}.

    We start with some definitions.
    Let $\sigmaapx \in \{0,\tilde{1},2\}^{\cw}$ be a tuple with $\cw$ entries,
    each taking values over the set $\{0,\tilde{1},2\}$.
    We call any such tuple a \emph{pseudosignature} and we say that the signature $\sigma \in \{0,1,2\}^{\cw}$
    is \emph{compatible} with the pseudosignature $\sigmaapx$ if for all $w \in [\cw]$ it holds that
    (i) if $\sigmaapx(w) = 0$ then $\sigma(w) = 0$,
    (ii) if $\sigmaapx(w) = \tilde{1}$ then $\sigma(w) \in \{0,1\}$, and
    (iii) if $\sigmaapx(w) = 2$ then $\sigma(w) = 2$.
    Let for $H' \in \mathcal{H}$, $\DPapx_{H'}[\cdot,\cdot]$ be the table
    which in entry $\DPapx_{H'}[k,\sigmaapx]$ contains the number of partial matchings
    in $H'$ of size $k$ whose signatures are compatible with $\sigmaapx$.
    Our main idea is to first compute the tables $\DPapx_{H_1}[\cdot,\cdot]$ and $\DPapx_{H_2}[\cdot,\cdot]$,
    then from those compute the table $\DPapx_H[\cdot,\cdot]$,
    and lastly use the latter to populate the table $\DPt_H[\cdot,\cdot]$.

    \proofsubparagraph{Step 1.}
    Let $H_i \in \{H_1,H_2\}$.
    For $\lambda \in [0,\cw]$ we say that $\sigmamixed \in \{0,\tilde{1},2\}^{\lambda} \times \{0,1,2\}^{\cw-\lambda}$ is
    a \emph{$\lambda$-mixed signature}.
    We say that the signature $\sigma \in \{0,1,2\}^{\cw}$ is \emph{compatible} with the $\lambda$-mixed signature $\sigmamixed$
    if the following hold:
    \begin{itemize}
        \item For all $w \in [\lambda]$ it holds that
            (i) if $\sigmamixed(w) = 0$ then $\sigma(w) = 0$,
            (ii) if $\sigmamixed(w) = \tilde{1}$ then $\sigma(w) \in \{0,1\}$, and
            (iii) if $\sigmamixed(w) = 2$ then $\sigma(w) = 2$.

        \item For all $w \in [\lambda+1,\cw]$, $\sigma(w) = \sigmamixed(w)$.
    \end{itemize}
    We proceed inductively.
    For $\lambda \in [0,\cw]$ let $\DPapx_{H_i}^{\lambda}[\cdot,\cdot]$ denote a table taking values over
    $[0,n] \times \{0,\tilde{1},2\}^{\lambda} \times \{0,1,2\}^{\cw-\lambda}$,
    where intuitively $\DPapx_{H_i}^{\lambda}[k,\sigmamixed]$ is equal to the number of partial matchings of ${H_i}$ of size $k$
    whose signature is compatible with $\sigmamixed$.
    For $\lambda=0$, this is precisely the table $\DPt_{H_i}[\cdot,\cdot]$,
    while for $\lambda=\cw$ this will be the table $\DPapx_{H_i}[\cdot,\cdot]$.

    It suffices to explain how to compute $\DPapx_{H_i}^{\lambda+1}[\cdot,\cdot]$ from the table
    $\DPapx_{H_i}^{\lambda}[\cdot,\cdot]$ in time $\sO(3^{\cw})$, and repeat the whole procedure $\cw$ times.
    Let $k \in [0,n]$ and $\sigmamixed \in \{0,\tilde{1},2\}^{\lambda+1} \times \{0,1,2\}^{\cw-(\lambda+1)}$
    be a $(\lambda+1)$-mixed signature.
    Consider two cases.
    If $\sigmamixed(\lambda+1) \in \{0,2\}$,
    then set $\DPapx_{H_i}^{\lambda+1}[k,\sigmamixed] = \DPapx_{H_i}^{\lambda}[k,\sigmamixed]$.
    Otherwise, it holds that $\sigmamixed(\lambda+1) = \tilde{1}$,
    and let $\sigmamixed', \sigmamixed'' \in \{0,\tilde{1},2\}^{\lambda} \times \{0,1,2\}^{\cw-\lambda}$
    denote the $\lambda$-mixed signatures obtained from $\sigmamixed$ by substituting the $(\lambda+1)$-th
    entry with $1$ and $0$ respectively, that is,
    $\sigmamixed' = \sigmamixed [(\lambda+1) \mapsto 0]$ and
    $\sigmamixed'' = \sigmamixed [(\lambda+1) \mapsto 1]$.
    In that case, set
    $\DPapx_{H_i}^{\lambda+1}[k,\sigmamixed] = \DPapx_{H_i}^{\lambda}[k,\sigmamixed'] + \DPapx_{H_i}^{\lambda}[k,\sigmamixed'']$.

    \proofsubparagraph{Step 2.}
    Utilizing the tables $\DPapx_{H_1}[\cdot,\cdot]$ and $\DPapx_{H_2}[\cdot,\cdot]$ constructed in Step 1,
    along with making use of the FFT technique introduced by Cygan and Pilipczuk~\cite{tcs/CyganP10},
    we compute in time $\sO(3^{\cw})$ the table $\DPapx_H[\cdot,\cdot]$,
    also indexed over $[0,n] \times \{0,\tilde{1},2\}^{\cw}$.
    Initially we set every entry of said table to $0$.

    In the following, fix $k \in [0,n]$ and a subset $A \subseteq [\cw]$.
    Let $\mathcal{S}_A \subseteq \{0,\tilde{1},2\}^{\cw}$ denote the set of pseudosignatures
    such that $\sigmaapx \in \mathcal{S}_A$ iff for all $w \in [\cw]$, $\sigmaapx(w) = \tilde{1} \iff w \in A$.
    We show how to compute all entries $\DPapx_H[k,\sigmaapx]$ with $\sigmaapx \in \mathcal{S}_A$
    in time $\sO(2^{\cw - |A|})$.
    Since $\sum_{A \subseteq [\cw]} 2^{\cw - |A|} = 3^{\cw}$ and $k$ takes $n+1$ different
    values, this suffices for the stated time bound.
    Let $w_1, \ldots, w_{\cw-|A|}$ denote in increasing order the non $\tilde{1}$-valued coordinates of a pseudosignature in $\mathcal{S}_A$,
    that is, for all $j \in [\cw-|A|]$ and $\sigmaapx \in \mathcal{S}_A$,
    we have that $\sigmaapx(w_j) \neq \tilde{1}$.

    We first fix $k_1,k_2 \in [0,k]$ such that $k_1+k_2=k$.
    Notice that there are $\bO(n)$ such pairs of $k_1,k_2$.
    Intuitively, $k_i$ represents the number of vertices of the partial matching of $H_i$ that account towards the
    total size of the partial matching in $H$.
    Next we define a \emph{binary encoding} $\bin_A \colon \mathcal{S}_A \to [0,2^{(\cw-|A|)}-1]$
    of every pseudosignature in $\mathcal{S}_A$ such that for all $\sigmaapx \in \mathcal{S}_A$
    the $j$-th bit of the binary representation of $\bin_A(\sigmaapx)$, denoted by $\bin_{A,j}(\sigmaapx)$,
    is equal to
    \[
        \bin_{A,j}(\sigmaapx) =
            \begin{cases}
                0   &\text{if $\sigmaapx(w_j) = 0$,}\\
                1   &\text{if $\sigmaapx(w_j) = 2$.}
            \end{cases}
    \]
    For every entry $\DPapx_{H_i}[k_i,\sigmaapx]$ with $\sigmaapx \in \mathcal{S}_A$ we introduce a monomial
    $\DPapx_{H_i}[k_i,\sigmaapx] \cdot x^{\bin_A(\sigmaapx)}$,
    and we set $P_{A,H_i,k_i}$ to denote the polynomial comprised of all such monomials,
    that is,
    \[
        P_{A,H_i,k_i} = \sum_{\sigmaapx \in \mathcal{S}_A} \DPapx_{H_i}[k_i,\sigmaapx] \cdot x^{\bin_A(\sigmaapx)}.
    \]
    Notice that the degree of $P_{A,H_i,k_i}$ is at most $2^{\cw-|A|}$, while all coefficients (i.e., the number of partial solutions)
    are upper-bounded by $2^n$.
    Let $P_{A,H_i,k_i}^{r_i}$ denote the restriction of $P_{A,H_i,k_i}$ to monomials whose degree has
    \emph{exactly} $r_i$ out of the $\cw-|A|$ bits being $1$.
    Fix $r_1,r_2 \in [0,\cw-|A|]$ and multiply polynomials $P_{A,H_1,k_1}^{r_1}$ and $P_{A,H_2,k_2}^{r_2}$ in
    $\sO(2^{\cw - |A|})$ time using FFT.%
    \footnote{Recall that using FFT we can multiply two polynomials of degree $n$ whose coefficients are upper-bounded by $2^{\bO(n)}$ in
    $\bO(n \log n)$ time~\cite{issac/Moenck76}.}
    We iterate over all $\sigmaapx \in \mathcal{S}_A$ having exactly $r_1+r_2$ bits of
    $\bin_A(\sigmaapx)$ being $1$ and check what coefficient, say $s$, stands in front
    of $x^{\bin_A(\sigmaapx)}$ in the resulting polynomial.
    We update the value of $\DPapx_H[k,\sigmaapx]$ by adding to its previous value $+s$.
    We repeat over all $r_1,r_2 \in [0,\cw-|A|]$, $k_1,k_2 \in [0,k]$, and $k \in [0,n]$,
    needing in total $\sO(2^{\cw-|A|})$ time.
    Finally, we fill the whole table $\DPapx_H[\cdot,\cdot]$ in time $\sO(3^{\cw})$ by repeating
    over all $A \subseteq [\cw]$.

    \proofsubparagraph{Step 3.}
    Given $\DPapx_H[\cdot,\cdot]$, we populate the values of $\DPt_H[\cdot,\cdot]$
    by repeating a procedure analogous to the one described in Step 1.
    The total running time is once again $\sO(3^{\cw})$.

    \begin{lemmarep}[\appsymb]\label{lemma:induced:cw:union-correctness}
        Let $H = H_1 \oplus H_2$.
        Assume that for all $i \in \{1,2\}$, $k \in [0,n]$, and $\sigma \in \{0,1,2\}^{\cw}$,
        $\DPt_{H_i}[k,\sigma]$ is equal to the number of partial matchings $S$ of $H_i$ such that $|S|=k$ and $\sgn_{H_i}(S)=\sigma$.
        Then, $\DPt_H[k,\sigma]$ is equal to the number of partial matchings $S$ of $H$ such that $|S|=k$ and $\sgn_H(S) = \sigma$.
    \end{lemmarep}

    \begin{proof}
        It is easy to see that, for $H_i \in \{H_1,H_2\}$,
        for all $k \in [0,n]$ and $\sigmaapx \in \{0,\tilde{1},2\}^{\cw}$,
        $\DPapx_{H_i}[k,\sigmaapx]$ is equal to the number of partial matchings $S$ of $H_i$ such that $|S|=k$ and
        $\sgn_{H_i}(S)$ is compatible with $\sigmaapx$.

        We proceed to show the correctness of Step 2 of the construction.
        Let $S$ be a partial matching of $H$ with $\sigma = \sgn_H(S)$.
        Moreover, let for $i \in \{1,2\}$,
        $S_i = S \cap V(H_i)$ with $\sigma_i = \sgn_{H_i}(S_i)$.

        \begin{claim}\label{claim:induced:cw:union-correctness}
            For all $w \in [\cw]$ it holds that
            \begin{itemize}
                \item $\sigma(w)=0 \iff \sigma_1(w)=\sigma_2(w)=0$,
                \item $\sigma(w)=1 \iff (\sigma_1(w),\sigma_2(w)) \in \{(0,1),(1,0),(1,1)\}$,
                \item $\sigma(w)=2 \iff (\sigma_1(w),\sigma_2(w)) \in \{(0,2),(2,0)\}$.
            \end{itemize}
        \end{claim}

        \begin{claimproof}
            Let $i \in \{1,2\}$.
            Notice that for all $v \in S_i$ it holds that
            $\lab_{H_i}(v) = \lab_H(v)$,
            thus $S \cap \lab^{-1}_H(w) =
            (S_1 \cap \lab^{-1}_{H_1}(w)) \cup (S_2 \cap \lab^{-1}_{H_2}(w))$.
            Furthermore, it holds that $\deg_{H_i[S_i]}(v) = \deg_{H[S]}(v)$.

            For the first item, due to the previous paragraph we have that
            if $\sigma(w) = 0$, that is, $S \cap \lab^{-1}_H(w) = \varnothing$,
            then $S_1 \cap \lab^{-1}_{H_1}(w) = S_2 \cap \lab^{-1}_{H_2}(w) = \varnothing$,
            that is, $\sigma_1(w) = \sigma_2(w) = 0$.
            It can be easily shown that the converse is also true.

            As for the second item, assume that $\sigma(w) = 1$, that is, $S \cap \lab^{-1}_H(w) \neq \varnothing$
            and for all of its vertices $v$, it holds that $\deg_{H[S]}(v)=1$.
            In that case, it follows that for all $v \in (S_1 \cap \lab^{-1}_{H_1}(w)) \cup (S_2 \cap \lab^{-1}_{H_2}(w))$,
            $\deg_{H_i[S_i]}(v)=1$, with at most one $S_i \cap \lab^{-1}_{H_i}(w)$ being empty.
            Conversely, assuming that $(\sigma_1(w),\sigma_2(w)) \in \{(0,1),(1,0),(1,1)\}$, it follows that
            for every vertex $v$ of the non-empty set $S \cap \lab^{-1}_H(w)$ it holds that $\deg_{H[S]}(v) = 1$.

            Lastly, assume that $\sigmaapx(w)=2$,
            which implies that $S \cap \lab^{-1}(w) = \{v\}$ and $\deg_{H[S]}(v)=0$.
            Assume that $v \in S_{i_1}$ and let $S_{i_2}$ denote the other set among $\{S_1,S_2\}$.
            Then, it follows that $S_{i_1} \cap \lab^{-1}_{H_{i_1}}(w) = \{v\}$ with $\deg_{H_{i_1}[S_{i_1}]}(v)=0$
            and $S_{i_2} \cap \lab^{-1}_{H_{i_2}}(w) = \varnothing$,
            implying that $\sigma_{i_1}(w) = 2$ and $\sigma_{i_2}(w) = 0$.
            It can be easily shown that the converse is also true.
        \end{claimproof}

        Recall that $\sigma \in \{0,1,2\}^{\cw}$ is compatible with $\sigmaapx \in \{0,\tilde{1},2\}^{\cw}$
        if for all $w \in [\cw]$ we have that
        (i) if $\sigmaapx(w) = 0$, then $\sigma(w) = 0$,
        (ii) if $\sigmaapx(w) = \tilde{1}$, then $\sigma(w) \in \{0,1\}$, and
        (iii) if $\sigmaapx(w) = 2$, then $\sigma(w) = 2$.
        Fix $k \in [0,n]$ and $\sigmaapx \in \{0,\tilde{1},2\}^{\cw}$, with $\sigmaapx \in \mathcal{S}_A$
        for some $A \subseteq [\cw]$.
%--------------- EDO EXO KANEI PATENTA ---------------%
        \iflncs
        The previous claim
        \else
        \cref{claim:induced:cw:union-correctness}
        \fi
        implies that the number of partial matchings of $H$ of size $k$
        and signature compatible with $\sigmaapx$ is equal to
        the number of partial matchings of $H_1$ of size $k_1$ and signature compatible with $\sigmaapx_1$
        times the number of partial matchings of $H_2$ of size $k_2$ and signature compatible with $\sigmaapx_2$,
        summing over all $k_1+k_2=k$ and $\sigmaapx_1,\sigmaapx_2 \in \mathcal{S}_A$ such that for all $w \in [\cw]$
        (i) if $\sigmaapx(w) = 0$, then $\sigmaapx_1(w) = \sigmaapx_2(w) = 0$,
        and (ii) if $\sigmaapx(w) = 2$, then $(\sigmaapx_1(w),\sigmaapx_2(w)) \in \{(0,2),(2,0)\}$.
        Notice that the latter condition imposed on the pseudosignatures $\sigmaapx_1,\sigmaapx_2$ is equivalent to
        demanding that the number $r_1$ of $1$-bits of $\bin_A(\sigmaapx_1)$ plus the number $r_2$ of $1$-bits of $\bin_A(\sigmaapx_2)$
        is equal to the number of $1$-bits of $\bin_A(\sigmaapx)$.

        In that case, for fixed $r_1,r_2 \in [0,\cw-|A|]$, it holds that the coefficient of a
        monomial $x^{\bin_A(\sigmaapx)}$ of $P_{A,H_1,k_1}^{r_1} \cdot P_{A,H_2,k_2}^{r_2}$,
        where $\bin_A(\sigmaapx)$ has exactly $r_1+r_2$ bits set to $1$,
        is equal to the number of partial matchings of $H$ of size $k = k_1+k_2$ and signature compatible
        with $\sigmaapx \in \mathcal{S}_A$
        that occur due to the union of partial matchings of $H_1$ of size $k_1$ and signature compatible with $\sigmaapx_1$
        and partial matchings of $H_2$ of size $k_2$ and signature compatible with $\sigmaapx_2$,
        where $\sigmaapx_1, \sigmaapx_2 \in \mathcal{S}_A$, $\bin_A(\sigmaapx_i)$ has exactly $r_i$ bits set to $1$,
        and $\bin_A(\sigmaapx)$ has exactly $r_1+r_2$ bits set to $1$.
        Since we iterate over all $r_1,r_2 \in [0,\cw-|A|]$ and $k_1,k_2 \in [0,k]$,
        it holds that indeed $\DPapx_H[k,\sigmaapx]$ contains the number of partial matchings of $H$ of size $k$
        and signature that is compatible with $\sigmaapx$, for all $\sigmaapx \in \mathcal{S}_A$.

        As for Step 3, given all the previous discussion,
        it is easy to see that it indeed constructs the table $\DPt_H[\cdot,\cdot]$
        so that $\DPt_H[k,\sigma]$ contains the number of partial matchings of $H$ of size $k$ and signature $\sigma$.
        This completes the proof.
    \end{proof}

    \proofsubparagraph{Wrapping up.}
    We argue about the correctness of the algorithm by a bottom-up induction on the clique-width $\cw$-expression of~$G$.
    We claim that for all $H \in \mathcal{H}$, $k \in [0,k]$, and $\sigma \in \{0,1,2\}^{\cw}$,
    $\DPt_H[k,\sigma]$ is equal to the number of partial matchings in $H$ of size $k$ and signature $\sigma$.
    Let $H = i(v)$ be a singleton graph for some label $i \in [\cw]$.
    There are only two possible partial matchings in $H$, the first being $\{v\}$ and the second being $\varnothing$.
    Notice that then by \cref{eq:DPH-singleton}, the statement holds for singleton graphs.
    For the Join, Relabel, and Union nodes the statement holds due to
    \cref{lemma:induced:cw:join-correctness,lemma:induced:cw:relabel-correctness,lemma:induced:cw:union-correctness}.
    As for the algorithm's running time, notice that for any $\cw$-labeled graph $H \in \mathcal{H}$ arising in~$\psi$,
    the table $\DPt_H[\cdot,\cdot]$ can be filled in time~$\sO(3^{\cw})$.
    As the number of such graphs is~$\bO(|\psi|)$, this implies that our algorithm runs in time $\sO(3^{\cw})$.
    % \todo[inline]{``Finally notice...'' This ``notice'' relies on \cref{claim:induced:cw:well_defined} so it's not just an observation.
    % I think it would be nice for the paper to, at least in appendix, prove that the number of induced matching is \textbf{exactly}
    % the table entry you claim by proving both inequalities, even if it's only a couple of lines.}
    Finally, notice that for all $H \in \mathcal{H}$,
    the number of induced matchings of $H$ of size $\ell$ is equal to the sum of $\DPt_H[2\ell,\sigma]$
    over all $\sigma \in \{0,1\}^{\cw}$.
    This concludes the proof of \cref{thm:induced:cw}.
\end{proof}

\section{Acyclic Matching}\label{sec:acyclic}

In this section we consider the {\AcyclicM} problem.
Chaudhary and Zehavi~\cite{siamdm/ChaudharyZ25} recently gave an $\sO(6^{\tw})$ algorithm for this problem,
where $\tw$ denotes the treewidth of the input graph.
We first notice that, with some slightly more careful analysis of its running time,
we can show that the algorithm in fact runs in time $\sO(5^{\tw})$.
Our main contribution however is to show in \cref{subsec:acyclic:lb} that this is optimal under the pw-SETH,
even for the more restricted parameterization by pathwidth.
Lastly, in \cref{subsec:acyclic:cw} we employ the rank-based approach introduced by
Bodlaender, Cygan, Kratsch, and Nederlof~\cite{iandc/BodlaenderCKN15}
and adapted to the parameterization by clique-width by
Bergougnoux and Kant\'e~\cite{tcs/BergougnouxK19}
to develop a $2^{\bO(\cw)} n^{\bO(1)}$ algorithm for {\AcyclicM} parameterized by the clique-width $\cw$
of the input graph.

\begin{theoremrep}[\appsymb]\label{thm:acyclic:tw_algo}
    Given a graph $G$ along with a nice tree decomposition of $G$ of width $\tw$,
    there is an algorithm that computes the size of the maximum acyclic matching in $G$
    in time $\sO(5^{\tw})$.
    The algorithm is randomized, cannot give false positives,
    and may give false negatives with probability at most $1/3$.
\end{theoremrep}

\begin{proof}
    The algorithm is identical to the one of Chaudhary and Zehavi~\cite{siamdm/ChaudharyZ25},
    with the only difference being that we analyze its running time slightly more carefully.
    In the following we only sketch the main ideas and argue about the claimed running time.

    This is a standard DP algorithm over the nodes of the tree decomposition.
    In order to deal with the acyclicity constraint, the authors employ the Cut \& Count technique~\cite{talg/CyganNPPRW22},
    while fast subset convolution is used to speed up the computation in the Join nodes of the tree decomposition (see~\cite[Chapter~11]{books/CyganFKLMPPS15}).
    Observe that the $\sO(6^{\tw})$ running time claimed in~\cite{siamdm/ChaudharyZ25} is solely due to the Join nodes,
    as for any other type of node $\sO(5^{\tw})$ time suffices to populate all the entries of the corresponding table.

    The key observation is that computing all entries of the table of a Join node can be done in time $\sO(5^{\tw})$.
    Following the same notation, let $\mathcal{B}_x$ denote a bag of a Join node $x$ of the tree decomposition.
    We aim to populate the table $\mathcal{A}_x[a,b,c,d,w,s]$, where $a,c \in [0,n]$, $b \in [0,n-1]$,
    $w \in [0,12n^2]$,
    and $d$ and $s$ are colorings such that $d \colon \mathcal{B}_x \to \{0,1,2\}$ and $s \colon \mathcal{B}_x \to \{0,l,r\}$
    with $d(v) = 0 \iff s(v) = 0$ for all $v \in \mathcal{B}_x$; we call such colorings $d$ and $s$ \emph{compatible}.
    Notice that the table $\mathcal{A}_x$ has $\sO(5^{\tw})$ entries.

    The authors prove that for fixed values of $a,b,c,w,s$,
    and for all colorings $d$ such that $d^{-1}(0) = R \subseteq \mathcal{B}_x$,
    one can compute the entries $\mathcal{A}_x [a,b,c,d,w,s]$ in time $2^{|\mathcal{B}_x \setminus R|} n^{\bO(1)}$.
    Consequently, summing over all subsets $R \subseteq \mathcal{B}_x$ and noticing that
    $\sum_{R \subseteq \mathcal{B}_x} 2^{|\mathcal{B}_x \setminus R|} = 3^{|\mathcal{B}_x|} \le 3^{\tw+1}$
    implies that for fixed values of $a,b,c,w,s$, the total time to compute the entries
    $\mathcal{A}_x [a,b,c,d,w,s]$ for \emph{all} colorings $d$ is $\sO(3^{\tw})$.
    Since $a,c \in [0,n]$, $b \in [0,n-1]$, $w \in [0,12n^2]$, and for every $d$ there are at most $2^{\tw}$
    compatible colorings $s$, it follows that there are $\sO(2^{\tw})$ ways to fix our parameters,
    resulting in a total time of $\sO(6^{\tw})$ to compute all entries of the table $\mathcal{A}_x$.
    Notice however that for a coloring $d$ where $d^{-1}(0) = R \subseteq \mathcal{B}_x$,
    it holds that there are at most $2^{|\mathcal{B}_x \setminus R|}$ compatible colorings $s$.
    Since $\sum_{R \subseteq \mathcal{B}_x} 2^{|\mathcal{B}_x \setminus R|} \cdot 2^{|\mathcal{B}_x \setminus R|} =
    5^{|\mathcal{B}_x|} \le 5^{\tw+1}$
    it holds that for fixed values of $a,b,c,w$,
    the total time to compute the entries of $\mathcal{A}_x [a,b,c,d,w,s]$
    for \emph{all compatible} colorings $d$ and $s$ is $\sO(5^{\tw})$,
    and thus the total time to compute all entries of the table $\mathcal{A}_x$ is $\sO(5^{\tw})$.
\end{proof}

\subsection{Lower Bound}\label{subsec:acyclic:lb}

Here we state the following lower bound for \AcyclicM.

\begin{theoremrep}[\appsymb]\label{thm:acyclic:lb}
    For any constant $\varepsilon > 0$,
    there is no $\sO((5-\varepsilon)^\pw)$ algorithm deciding \AcyclicM,
    where $\pw$ denotes the pathwidth of the input graph,
    unless the pw-SETH is false.
\end{theoremrep}

\begin{proof}
    We present a reduction from the $4$-CSP-$5$ problem of \cref{cor:weird} to \AcyclicM.
    We are given a \textsc{CSP} instance $\psi$ whose variables take values from $[5]$,
    a partition of its variables into two sets $V_1, V_2$,
    a path decomposition $B_1, \ldots, B_t$ of $\psi$ of width $p$
    such that each bag contains $\bO(\log p)$ variables of $V_2$,
    and an injective function $b$ mapping each constraint to a bag that contains its variables.
    We want to construct an {\AcyclicM} instance $(G,\ell)$ with $\pw(G) = p + o(p)$,
    such that if $\psi$ is satisfiable, then $G$ has an acyclic matching of size $\ell$,
    while if $G$ has an acyclic matching of size $\ell$,
    $\psi$ admits a monotone satisfying multi-assignment which is consistent for $V_2$.
    We assume that the variables of $\psi$ are numbered $X = \{ x_1,\ldots, x_n \}$,
    the constraints are $c_1, \ldots, c_m$,
    and that the given path decomposition is nice.
    We construct $G$ as follows.

    \proofsubparagraph{XOR Gadget.}
    Given two vertices $v_1$ and $v_2$, when we say that we connect them via a \emph{XOR gadget} $\hat{O}(v_1,v_2)$
    we first add the edge $\{v_1,v_2\}$,
    then introduce two vertices $p^{v_1,v_2}_1$ and $p^{v_1,v_2}_2$ called the \emph{private vertices of the XOR gadget},
    and lastly add edges so that $v_1$, $v_2$, and the private vertices of the gadget form a cycle;
    for an illustration see \cref{fig:acyclic:lb_or_gadget}.
    We say that $v_1,v_2,p^{v_1,v_2}_1,p^{v_1,v_2}_2$ are the \emph{vertices of the XOR gadget $\hat{O}(v_1,v_2)$},
    while the \emph{edges of a XOR gadget} are the edges whose endpoints both are among its vertices.

    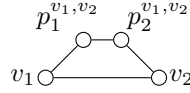
\begin{figure}[htb]
        \centering
        \begin{tikzpicture}[scale=1, transform shape]

        %%%%%%%%%% vertices and text

        \node[vertex] (u) at (1,5) {};
        \node[] () at (0.7,5) {$v_1$};

        \node[vertex] (x1) at (1.5,5.5) {};
        \node[] () at (1.3,5.8) {$p^{v_1,v_2}_1$};
        \node[vertex] (x2) at (2,5.5) {};
        \node[] () at (2.5,5.8) {$p^{v_1,v_2}_2$};

        \node[vertex] (v) at (2.5,5) {};
        \node[] () at (2.8,5) {$v_2$};

        %%%%%%%%% edges / arcs

        \draw[] (u)--(x1)--(x2)--(v)--(u);

        \end{tikzpicture}
        \caption{XOR gadget $\hat{O}(v_1,v_2)$.}
        \label{fig:acyclic:lb_or_gadget}
    \end{figure}

    \proofsubparagraph{Root vertex.}
    We start our construction by introducing a vertex $r$, which we refer to as the \emph{root vertex}.
    Furthermore, we attach a leaf $r'$ to $r$.

    \proofsubparagraph{Block and Variable Gadgets.}
    For every variable $x_i$ and every bag with $x_i \in B_j$, where $j \in [j_1,j_2]$,
    construct a \emph{block gadget} $\hat{B}_{i,j}$.
    Here we consider two cases, depending on whether $x_i \in V_1$ or $x_i \in V_2$.

    If $x_i \in V_1$, then we construct the block gadget as depicted in \cref{fig:acyclic:lb_block_gadget}.
    In order to do so, we introduce vertices $a, a', \chi_1, \chi_2, y_1, y_2, b_1, b_2$.
    Then, we attach leaves $l_1$ and $l_2$ to $\chi_1$ and $\chi_2$ respectively.
    Next, we add edges from $a$ to $\chi_1$ and $y_1$, from $a'$ to $\chi_2$ and $y_2$,
    from the root vertex $r$ to $\chi_1$ and $\chi_2$, as well as the edge $\{b_1,b_2\}$.
    Finally, we add the XOR gadgets $\hat{O}(a,b_1)$, $\hat{O}(a',b_2)$, $\hat{O}(\chi_1,\chi_2)$,
    and $\hat{O}(y_1,y_2)$.
    Next, for $j \in [j_1,j_2-1]$, we serially connect the block gadgets $\hat{B}_{i,j}$ and $\hat{B}_{i,j+1}$ so that the vertex $a'$ of $\hat{B}_{i,j}$ is the vertex $a$ of $\hat{B}_{i,j+1}$,
    thus resulting in the ``path'' $\hat{P}_i$ consisting of such serially connected block gadgets, called a \emph{variable gadget}.
    For an illustration see \cref{fig:acyclic:lb_serially_connected}.
    Intuitively, for any optimal matching $M$,
    it will hold that either $a, a' \notin V_M$ or $a, a' \in V_M$.
    In the latter case, it will hold that $|V_M \cap \{\chi_1,\chi_2\}| = |V_M \cap \{y_1,y_2\}| = 1$,
    resulting in $4$ choices for a total of $5$ possible choices per path gadget;
    we map each choice to an assignment where $x_i$ receives a value in $[5]$.

    \begin{figure}[ht]
        \centering
          \begin{subfigure}[b]{0.4\linewidth}
          \centering
            \begin{tikzpicture}[scale=0.75, transform shape]

            %%%%%%%%%% vertices and text

            \node[vertex] (a) at (0,5) {};
            \node[] () at (0,5.3) {$a$};

            \node[gray_vertex] (x1) at (1.5,6) {};
            \node[] () at (1.5,6.3) {$\chi_1$};

            \node[gray_vertex] (x2) at (4.5,6) {};
            \node[] () at (4.5,6.3) {$\chi_2$};

            \node[vertex] (y1) at (1.5,4) {};
            \node[] () at (1.5,4.3) {$y_1$};

            \node[vertex] (y2) at (4.5,4) {};
            \node[] () at (4.5,4.3) {$y_2$};

            \node[vertex] (a') at (6,5) {};
            \node[] () at (6,5.3) {$a'$};

            \node[vertex] (b1) at (2.5,2.5) {};
            \node[] () at (2.5,2.8) {$b_1$};
            \node[vertex] (b2) at (3.5,2.5) {};
            \node[] () at (3.5,2.8) {$b_2$};

            \node[gray_vertex] (r) at (3,7.5) {};
            \node[] () at (3,7.8) {$r$};

            %%%%%%%%% edges / arcs

            \draw[] (r)--(x1);
            \draw[] (r)--(x2);
            \draw[] (a)--(x1);
            \draw[] (x2)--(a');
            \draw[dashed] (x1)--(x2);
            \draw[] (a)--(y1);
            \draw[] (y2)--(a');
            \draw[dashed] (y1)--(y2);
            \draw[] (b1)--(b2);
            \draw[dashed] (a) edge [bend right] (b1);
            \draw[dashed] (a') edge [bend left] (b2);

            \end{tikzpicture}
            \caption{Block gadget.}
            \label{fig:acyclic:lb_block_gadget}
          \end{subfigure}
        \begin{subfigure}[b]{0.4\linewidth}
        \centering
            \begin{tikzpicture}[scale=0.5, transform shape]

            \node[vertex] (a) at (0,5) {};

            \node[gray_vertex] (x1) at (1.5,6) {};

            \node[gray_vertex] (x2) at (4.5,6) {};

            \node[vertex] (y1) at (1.5,4) {};

            \node[vertex] (y2) at (4.5,4) {};

            \node[vertex] (a') at (6,5) {};

            \node[vertex] (b1) at (2.5,2.5) {};
            \node[vertex] (b2) at (3.5,2.5) {};

            \begin{scope}[shift={(6,0)}]

                \node[gray_vertex] (x1') at (1.5,6) {};

                \node[gray_vertex] (x2') at (4.5,6) {};

                \node[vertex] (y1') at (1.5,4) {};

                \node[vertex] (y2') at (4.5,4) {};

                \node[vertex] (b1') at (2.5,2.5) {};
                \node[vertex] (b2') at (3.5,2.5) {};

                \node[vertex] (a'') at (6,5) {};

            \end{scope}

            \node[gray_vertex] (r) at (6,8.5) {};

            %%%%%%%%% edges / arcs

            \draw[] (r) edge [bend right] (x1);
            \draw[] (r) edge [bend right] (x2);
            \draw[] (r) edge [bend left] (x1');
            \draw[] (r) edge [bend left] (x2');
            \draw[] (a)--(x1);
            \draw[] (x2)--(a');
            \draw[dashed] (x1)--(x2);
            \draw[] (a)--(y1);
            \draw[] (y2)--(a');
            \draw[dashed] (y1)--(y2);
            \draw[] (b1)--(b2);
            \draw[dashed] (a) edge [bend right] (b1);
            \draw[dashed] (a') edge [bend left] (b2);
            \draw[] (a')--(x1');
            \draw[] (x2')--(a'');
            \draw[dashed] (x1')--(x2');
            \draw[] (a')--(y1');
            \draw[] (y2')--(a'');
            \draw[dashed] (y1')--(y2');
            \draw[] (b1')--(b2');
            \draw[dashed] (a') edge [bend right] (b1');
            \draw[dashed] (a'') edge [bend left] (b2');

            \end{tikzpicture}
            \caption{Serially connected block gadgets.}
            \label{fig:acyclic:lb_serially_connected}
          \end{subfigure}
        \caption{Case where $x_i \in V_1$. Gray vertices have a leaf attached, dashed edges denote XOR gadgets.}
        \label{fig:acyclic:lb}
        \end{figure}
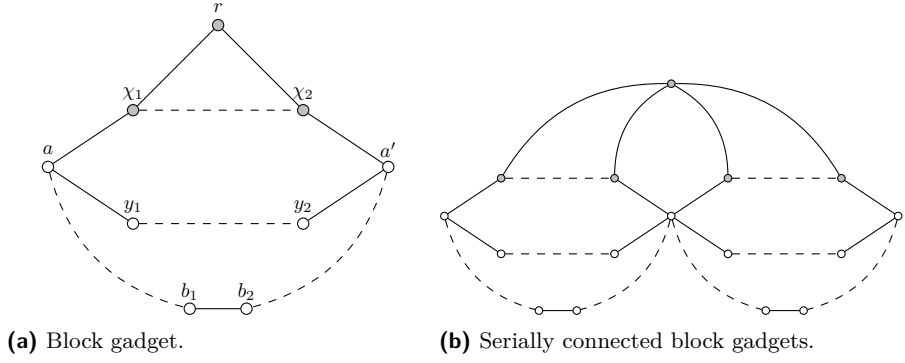

        As for the case where $x_i \in V_2$, then the block gadget $\hat{B}_{i,j}$
        is constructed as follows.
        We introduce $6$ vertices $u^{i,j}, u^{i,j}_1, u^{i,j}_2, u^{i,j}_3, u^{i,j}_4, u^{i,j}_5$,
        we add an edge from $u^{i,j}$ to all vertices in $\setdef{u^{i,j}_k}{k \in [5]}$,
        and for all $1 \le k_1 < k_2 \le 5$ we add a XOR gadget
        $\hat{O}(u^{i,j}_{k_1},u^{i,j}_{k_2})$.
        Next, for $j \in [j_1,j_2-1]$, for all \emph{distinct} $k_1,k_2 \in [5]$,
        we add a XOR gadget $\hat{O}(u^{i,j}_{k_1},u^{i,j+1}_{k_2})$,
        and we call the \emph{variable gadget} $\hat{P}_i$ this sequence of serially connected block gadgets.

        \proofsubparagraph{Constraint Gadget.}
        This gadget is responsible for determining constraint satisfaction,
        based on the choices made in the rest of the graph.
        Let $c$ be a constraint of $\psi$, where $b(c) = j$ and
        $c$ involves variables $x_{i_1},x_{i_2},x_{i_3}, x_{i_4} \in B_j$.
        Consider the set $\mathcal{S}_c$ of the at most $5^4$ satisfying assignments of $c$.
        We construct the \emph{constraint gadget} $\hat{C}_c$ as follows.
        \begin{itemize}
            \item For each $\sigma \in \mathcal{S}_c$ construct a vertex $v_{c,\sigma}$.

            \item For distinct $\sigma_1,\sigma_2 \in \mathcal{S}_c$,
            introduce a XOR gadget $\hat{O}(v_{c,\sigma_1},v_{c,\sigma_2})$.

            \item Introduce a vertex $v_c$ and add edges between $v_c$ and all vertices $\setdef{v_{c,\sigma}}{\sigma \in \mathcal{S}_c}$.

            \item For each $\sigma \in \mathcal{S}_c$ and $\alpha \in [4]$,
            consider the following cases:
            \begin{itemize}
                \item if $x_{i_\alpha} \in V_1$ and $\sigma(x_{i_\alpha}) = 1$,
                then add a XOR gadget between $v_{c,\sigma}$ and vertices $\{b_1, b_2, \chi_2, y_2\}$ of $\hat{B}_{i,j}$,

                \item if $x_{i_\alpha} \in V_1$ and $\sigma(x_{i_\alpha}) = 2$,
                then add a XOR gadget between $v_{c,\sigma}$ and vertices $\{b_1, b_2, \chi_1, y_2\}$ of $\hat{B}_{i,j}$,

                \item if $x_{i_\alpha} \in V_1$ and $\sigma(x_{i_\alpha}) = 3$,
                then add a XOR gadget between $v_{c,\sigma}$ and vertices $\{a, a',y_1, y_2\}$ of $\hat{B}_{i,j}$,

                \item if $x_{i_\alpha} \in V_1$ and $\sigma(x_{i_\alpha}) = 4$,
                then add a XOR gadget between $v_{c,\sigma}$ and vertices $\{b_1, b_2, \chi_2, y_1\}$ of $\hat{B}_{i,j}$,

                \item if $x_{i_\alpha} \in V_1$ and $\sigma(x_{i_\alpha}) = 5$,
                then add a XOR gadget between $v_{c,\sigma}$ and vertices $\{b_1, b_2, \chi_1, y_1\}$ of $\hat{B}_{i,j}$,

                \item if $x_{i_\alpha} \in V_2$,
                then add a XOR gadget between $v_{c,\sigma}$ and all vertices in $\setdef{u^{i,j}_k}{k \in [5] \setminus \{ \sigma(x_{i_\alpha}) \}}$.
            \end{itemize}
        \end{itemize}

        This completes the construction of $G$.
        We set $L$ to be the number of XOR gadgets introduced in $G$,
        and we set $\ell = 1 + L + 2L_1 + L_2 + m$,
        where $L_i$ denotes the number of block gadgets constructed due to variables belonging to $V_i$,
        and $m$ is the number of constraints of $\psi$.
        In that case, $(G,\ell)$ is the constructed instance of \AcyclicM.

        \begin{lemma}\label{lem:acyclic:lb:csp->acyclic}
            If $\psi$ is satisfiable,
            then $G$ has an acyclic matching of size $\ell$.
        \end{lemma}

        \begin{nestedproof}
            Let $\sigma \colon X \to [5]$ be a satisfying assignment for $\psi$, that is,
            $\sigma$ satisfies all constraints $c_1, \ldots, c_m$.
            We will construct a matching $M \subseteq E(G)$ of size $|M| \geq \ell$ such that $G[V_M]$ is a forest.
            \begin{itemize}
                \item First, add to $M$ the edge $\{r,r'\}$.

                \item Then, for every XOR gadget $\hat{O}(v_1,v_2)$ in $G$, include in $M$ the edge between
                its two private vertices $\{p^{v_1,v_2}_1,p^{v_1,v_2}_2\}$.

                \item Next, for every constraint gadget $\hat{C}_c$, since $\sigma$ is a satisfying assignment,
                there exists some vertex $v_{c,\sigma_c}$ with $\sigma_c$ being the restriction of $\sigma$
                to the variables involved in $c$.
                Add in $M$ the edge $\{v_c,v_{c,\sigma_c}\}$.

                \item For each $x_i \in V_2$ and $j \in [t]$ such that $x_i \in B_j$
                we select into $M$ the edge $\{u^{i,j}, u^{i,j}_k\}$, where $\sigma(x_i) = k$.

                \item Finally, for each $x_i \in V_2$ and $j \in [t]$ such that $x_i \in B_j$,
                we include in $M$ the following edges from the block gadget $\hat{B}_{i,j}$:
                \begin{itemize}
                    \item if $\sigma(x_i) = 1$, the edges $\{ y_1, a \}$ and $\{ \chi_1, l_1 \}$,

                    \item if $\sigma(x_i) = 2$, the edges $\{ y_1, a \}$ and $\{ \chi_2, l_2 \}$,

                    \item if $\sigma(x_i) = 3$, the edges $\{ b_1, b_2 \}$ and $\{ \chi_1, l_1 \}$,

                    \item if $\sigma(x_i) = 4$, the edges $\{ y_2, a' \}$ and $\{ \chi_1, l_1 \}$,

                    \item if $\sigma(x_i) = 5$, the edges $\{ y_2, a' \}$ and $\{ \chi_2, l_2 \}$.
                \end{itemize}
            \end{itemize}
            Consequently, $M$ is a matching of size $1 + L + m + L_2 + 2L_1 = \ell$.

            We claim that $G[V_M]$ is a forest.
            Consider any constraint gadget $\hat{C}_c$,
            and let $v_{c,\sigma_c}$ such that $\{v_c, v_{c,\sigma_c}\} \in M$.
            Notice that since $\sigma_c$ is the restriction of $\sigma$ to its involved variables,
            one can easily verify that none of the vertices that $v_{c,\sigma_c}$ has a XOR gadget with belongs to $V_M$.
            Consequently, the vertices of the constraint gadgets belonging to $V_M$ induce acyclic components.

            It remains to argue that the graph induced by $r$, $r'$, and the vertices of the block gadgets in $V_M$ is acyclic.
            Notice that this is indeed the case for the block gadgets corresponding to variables in $V_2$,
            thus in the following we focus on the block gadgets corresponding to variables in $V_1$.
            Let $x_i \in V_1$ and $j \in [t]$ such that $x_i \in B_j$.
            First, notice that the graph induced by the vertices of any single block gadget $\hat{B}_{i,j}$ in $V_M$ along with $\{r,r'\}$
            is acyclic and vertices $a$ and $a'$ are disconnected.
            Furthermore, for $j_1, j_2 \in [t]$ and distinct $x_{i_1},x_{i_2} \in V_1$,
            notice that the block gadgets $\hat{B}_{i_1,j_1}$ and $\hat{B}_{i_2,j_2}$
            are disconnected in $G[V_M] - r$, thus if $G[V_M]$ contains a cycle,
            then this cycle involves only block gadgets corresponding to a single variable $x_i$.
            Finally, it holds that any two sequential block gadgets (as in \cref{fig:acyclic:lb_serially_connected})
            are connected in $G - r$ via their single common vertex which appears as $a$ in the first and $a'$ in the second gadget.
            Consequently, any cycle in $G[V_M]$ involving at least two block gadgets $\hat{B}_{i,j_1}$ and $\hat{B}_{i,j_2}$ with $j_1 < j_2$
            must contain $r$, the vertex $a'$ of $\hat{B}_{i,j_1}$, and the vertex $a$ of $\hat{B}_{i,j_2}$.
            If $j_2 = j_1 + 1$, it is easy to verify that there exists only a single path in $G[V_M]$ that connects $r$ and
            vertex $a'$ of $\hat{B}_{i,j_1}$ (which coincides with the vertex $a$ of $\hat{B}_{i,j_2}$).
            Otherwise, i.e., when $j_2 > j_1+1$, given that the vertices $a$ and $a'$ of any single block gadget are disconnected in $G[V_M] - r$,
            it follows that vertex $a'$ of $\hat{B}_{i,j_1}$ and $a$ of $\hat{B}_{i,j_2}$
            are disconnected in $G[V_M] - r$, thus no cycle exists in $G[V_M]$.
            This completes the proof.
        \end{nestedproof}

        \begin{lemma}
            If $G$ has an acyclic matching of size $\ell$,
            then $\psi$ admits a monotone satisfying multi-assignment which is consistent for $V_2$.
        \end{lemma}

        \begin{nestedproof}
            We start with the following claim.

            \begin{claim}\label{claim:acyclic:lb:acyclic->csp:maximum_matching}
                There exists an acyclic matching $M_{\max}$ of $G$ of maximum size such that
                for any XOR gadget $\hat{O}(v_1,v_2)$ of $G$ it holds that $|\{v_1,v_2\} \cap V_{M_{\max}}| \le 1$
                and $\{p^{v_1,v_2}_1,p^{v_1,v_2}\} \in M_{\max}$.
            \end{claim}

            \begin{claimproof}
                Let $M_{\max}$ be a maximum acyclic matching of $G$, and assume that $G$ contains a XOR gadget $\hat{O}(v_1,v_2)$.
                We show that we can obtain an acyclic matching $M'_{\max}$ of $G$ of the same size
                such that $|\{v_1,v_2\} \cap V_{M'_{\max}}| \le 1$ and $\{p^{v_1,v_2}_1,p^{v_1,v_2}\} \in M'_{\max}$.
                To this end, we consider different cases depending on the vertices of the XOR gadget that belong to $V_{M_{\max}}$.

                If $v_1,v_2 \notin V_{M_{\max}}$, then $\{p^{v_1,v_2}_1,p^{v_1,v_2}_2\} \in M_{\max}$ as
                $M_{\max}$ is of maximum size.

                Now assume that $|\{v_1,v_2\} \cap V_{M_{\max}}| = 1$, and without loss of generality $v_1 \in V_{M_{\max}}$ (the other case is analogous).
                Let $e \in M_{\max}$ be the edge incident with $v_1$.
                If $e \neq \{v_1,p^{v_1,v_2}_1\}$, then it holds that $\{p^{v_1,v_2}_1,p^{v_1,v_2}_2\} \in M_{\max}$ as $M_{\max}$ is of maximum size.
                Otherwise it holds that $p^{v_1,v_2}_2 \notin V_{M_{\max}}$ and the acyclic matching $M'_{\max} = (M_{\max} \setminus \{e\}) \cup \{\{p^{v_1,v_2}_1,p^{v_1,v_2}_2\}\}$
                has the same size.

                Lastly, assume that $v_1,v_2 \in V_{M_{\max}}$.
                Notice that in that case, either $p^{v_1,v_2}_1 \notin V_{M_{\max}}$ or $p^{v_1,v_2}_2 \notin V_{M_{\max}}$, as otherwise $G[V_{M_{\max}}]$
                contains a cycle.
                If $\{v_1,v_2\} \in M_{\max}$, then $M'_{\max} = (M_{\max} \setminus \{\{v_1,v_2\}\}) \cup \{\{p^{v_1,v_2}_1,p^{v_1,v_2}_2\}\}$
                is an acyclic matching of the same size.
                Otherwise, it holds that for either $v_1$ or $v_2$ there exists an edge $e \in M_{\max}$ incident with it
                that does not belong to the XOR gadget $\hat{O}(v_1,v_2)$.
                Assume without loss of generality that $e$ is incident with $v_1$, and let $e' \in M_{\max}$ be the edge incident with $v_2$ belonging to $M_{\max}$.
                Then, the acyclic matching $M'_{\max} = (M_{\max} \setminus \{e'\}) \cup \{\{p^{v_1,v_2}_1,p^{v_1,v_2}_2\}\}$ is of the same size.
            \end{claimproof}

            Armed with the previous claim, we can now prove the existence of an acyclic matching in $G$ with specific properties.

            \begin{claim}\label{claim:acyclic:lb:acyclic->csp:matching_properties}
                There exists an acyclic matching $M$ of maximum size of $G$ such that,
                for any XOR gadget $\hat{O}(v_1,v_2)$ of $G$ it holds that $|\{v_1,v_2\} \cap V_M| \le 1$ and $\{p^{v_1,v_2}_1,p^{v_1,v_2}\} \in M$.
                Furthermore, it holds that
                \begin{itemize}
                    \item $\{r,r'\} \in M$,

                    \item for all constraint gadgets $\hat{C}_c$,
                    there exists $\sigma \in \mathcal{S}_c$ such that $\{v_c, v_{c,\sigma_c}\} \in M$,

                    \item for all block gadgets $\hat{B}_{i,j}$,
                    \begin{itemize}
                        \item if $x_i \in V_1$ then
                        $| \{ \{\chi_1,l_1\}, \{\chi_2,l_2\} \} \cap M| =
                        |\{\{ y_1, a \}, \{ y_2, a' \}, \{ b_1, b_2 \}\} \cap M| = 1$,

                        \item otherwise if $x_i \in V_2$ there exists $k \in [5]$ such that $\{u^{i,j}, u^{i,j}_k\} \in M$.
                    \end{itemize}
                \end{itemize}
            \end{claim}

            \begin{claimproof}
                Let $M_{\max}$ be an acyclic matching of $G$ of maximum size as in
%--------------- EDO EXO KANEI PATENTA ---------------%
                \iflncs
                the previous claim,
                \else
                \cref{claim:acyclic:lb:acyclic->csp:maximum_matching},
                \fi
                where $|M_{\max}| \ge \ell$ since by assumption $G$ has an acyclic matching of size at least $\ell$.
                Consider the following reduction rule.

                \proofsubparagraph{Rule~$(\dagger)$.}
                Let $\{u,v\} \in M_{\mathsf{in}}$, and let $l \neq v$ be a leaf attached to $u$ in $G$, while $v$ is not a leaf.
                Then, replace $M_{\mathsf{in}}$ with $M_{\mathsf{out}} = (M_{\mathsf{in}} \setminus \{\{u,v\}\}) \cup \{\{u,l\}\}$.

                It is easy to see that in Rule~$(\dagger)$ if $M_{\mathsf{in}}$ is an acyclic matching,
                then so is $M_{\mathsf{out}}$, while the size of the two matchings is the same.
                Let $M$ denote the acyclic matching obtained after exhaustively applying Rule~$(\dagger)$ in $M_{\max}$;
                since all vertices of $G$ have at most one attached leaf, such a matching is unique.
                Notice that $|M| \ge \ell$, while for any XOR gadget $\hat{O}(v_1,v_2)$ of $G$ it holds that $|\{v_1,v_2\} \cap V_M| \le 1$
                and $\{p^{v_1,v_2}_1,p^{v_1,v_2}\} \in M$.

                Notice that due to Rule~$(\dagger)$, if $r \in V_M$, then $\{r,r'\} \in M$.
                Furthermore, for any edge in $M \setminus \{\{r,r'\}\}$ that is not due to a XOR gadget
                it holds that either both of its endpoints belong to the vertices of a block gadget,
                or they both belong to the vertices of a constraint gadget;
                that is due to the fact that any connected pair of vertices from both block and constraint gadgets
                or from different block gadgets is connected via a XOR gadget.

                Notice that in the constraint gadget $\hat{C}_c$,
                all vertices $\setdef{v_{c,\sigma}}{\sigma \in \mathcal{S}_c}$ are pairwise connected via XOR gadgets,
                thus at most one of them belongs to $V_M$.
                Along with the previous paragraph, it follows that either $V_M$ contains no vertex from $\hat{C}_c$,
                or there exists $\sigma_c \in \mathcal{S}_c$ such that $\{v_c,v_{c,\sigma_c}\} \in M$.

                Now consider the block gadget $\hat{B}_{i,j}$, where $j \in [t]$ such that $x_i \in B_j$.
                Assume first that $x_i \in V_1$,
                and notice that if $\chi_1 \in V_M$, then $\{\chi_1,l_1\} \in M$ and $\chi_2 \notin V_M$ (similarly for $\chi_2$).
                Furthermore, one can easily verify that $|\{\{ y_1, a \}, \{ y_2, a' \}, \{ b_1, b_2 \}\} \cap M| \le 1$.
                If on the other hand $x_i \in V_2$, then
                since vertices $\setdef{u^{i,j}_k}{k \in [5]}$ are pairwise connected via XOR gadgets,
                it holds that either $V_M$ contains no vertex of $\hat{B}_{i,j}$,
                or there exists $k \in [5]$ such that $\{u^{i,j}, u^{i,j}_k\} \in M$.

                Recall that $|M| \ge \ell = 1 + L + 2L_1 + L_2 + m$,
                where $L$ is the number of XOR gadgets in $G$,
                $L_i$ is the number of block gadgets constructed due to variables belonging to $V_i$,
                and $m$ is the number of constraints of $\psi$.
                Due to the discussion so far and summing over all constraint and block gadgets,
                $M$ contains at most $\ell$ edges.
                Consequently, it follows that the bounds obtained in the previous paragraphs are tight for all constraint and block gadgets,
                which implies that the inclusions are as described in the statement and $\{r,r'\} \in M$.
            \end{claimproof}

            Now let $M$ be an acyclic matching of $G$ as in
%--------------- EDO EXO KANEI PATENTA ---------------%
            \iflncs
            the previous claim.
            \else
            \cref{claim:acyclic:lb:acyclic->csp:matching_properties}.
            \fi
            Consider a multi-assignment $\sigma \colon X \times [t] \to [5]$ over all pairs $x_i \in X$ and $j \in [t]$
            with $x_i \in B_j$.
            In particular, consider two cases.
            If $x_i \in V_2$ and $\{u^{i,j}, u^{i,j}_k\} \in M$, then $\sigma(x_i,j) = k \in [5]$.
            Otherwise (i.e., $x_i \in V_1$),
            if for the edges of $\hat{B}_{i,j}$ it holds that
            \begin{itemize}
                \item $\{y_1,a\},  \{\chi_1,l_1\} \in M$, then $\sigma(x_i,j) = 1$,
                \item $\{y_1,a\},  \{\chi_2,l_2\} \in M$, then $\sigma(x_i,j) = 2$,
                \item $\{b_1,b_2\} \in M$,                then $\sigma(x_i,j) = 3$,
                \item $\{y_2,a'\}, \{\chi_1,l_1\} \in M$, then $\sigma(x_i,j) = 4$,
                \item $\{y_2,a'\}, \{\chi_2,l_2\} \in M$, then $\sigma(x_i,j) = 5$.
            \end{itemize}
            Notice that $\sigma$ is well-defined due to
%--------------- EDO EXO KANEI PATENTA ---------------%
            \iflncs
            the previous claim.
            \else
            \cref{claim:acyclic:lb:acyclic->csp:matching_properties}.
            \fi

            \begin{claim}
                It holds that $\sigma$ is monotone-increasing, as well as consistent for all $x_i \in V_2$.
            \end{claim}

            \begin{claimproof}
                Regarding the consistency for $x_i \in V_2$,
                notice that it follows due to the XOR gadgets between vertices belonging to sequential
                block gadgets $\hat{B}_{i,j}$ and $\hat{B}_{i,j+1}$.
                In the following we argue about the monotonicity.
                Assume that $x_i \in V_1$ and $j \in [j_1,j_2-1]$,
                where $B_{j_1}$ and $B_{j_2}$ denote the first and last bag of the decomposition containing $x_i$ respectively.

                We first argue that if for $\hat{B}_{i,j}$ it holds that
                $\{y_2 , a'\} \in M$,
                then for $\hat{B}_{i,j+1}$ it holds that
                $\{y_2 , a'\} \in M$.
                Indeed, since vertices $a'$ and $a$ of the first and the latter block gadget coincide,
                it follows that for $\hat{B}_{i,j+1}$ we have $\{y_1, a\}, \{b_1,b_2\} \notin M$,
                and since $|\{\{ y_1, a \}, \{ y_2, a' \}, \{ b_1, b_2 \}\} \cap M| = 1$ for any block gadget,
                the statement follows.

                Now we argue that if for $\hat{B}_{i,j}$ it holds that
                $\{y_2 , a'\}, \{\chi_2,l_2\} \in M$,
                then for $\hat{B}_{i,j+1}$ it holds that
                $\{y_2 , a'\}, \{\chi_2,l_2\} \in M$.
                In this case, due to the previous paragraph we have that in $\hat{B}_{i,j+1}$,
                $\{y_2 , a'\} \in M$ as well as $a \in V_M$,
                implying that $\chi_1 \notin V_M$ (as otherwise $G[V_M]$ has a cycle),
                which in turn implies that $\{\chi_2 , l_2\} \in M$.
                Consequently, it holds that if $\sigma (x_i,j) = 5$ then $\sigma (x_i,j+1) = 5$,
                while in conjunction with the previous paragraph we get that if $\sigma (x_i,j) = 4$
                then $\sigma (x_i,j+1) \in \{4,5\}$.

                Next, it is easy to see that if in $\hat{B}_{i,j}$ we have that $\{b_1 ,b_2 \} \in M$,
                then in $\hat{B}_{i,j+1}$ we have that $\{y_1 , a \} \notin M$;
                indeed, in $\hat{B}_{i,j}$ since $b_2 \in V_M$ it holds that $a' \notin V_M$,
                however this implies that for $\hat{B}_{i,j+1}$  it holds that $a \notin V_M$.
                Consequently, if $\sigma (x_i,j) = 3$ then $\sigma (x_i,j+1) \in \{3,4,5\}$.

                Finally assume that in $\hat{B}_{i,j}$
                we have that $\{y_1 ,a\} , \{\chi_2 , l_2\} \in M$.
                In that case, it holds that in $\hat{B}_{i,j+1}$ we have $\chi_1 \notin V_M$
                (otherwise $G[V_M]$ contains a cycle),
                implying that $\{\chi_2, l_2\} \in M$.
                Consequently, we have that if $\sigma (x_i,j) = 2$ then
                $\sigma (x_i,j+1) \neq 1$.
            \end{claimproof}

            It remains to argue that $\sigma$ is a satisfying multi-assignment for $\psi$.
            Consider a constraint $c$ with $b(c)=j$,
            involving four variables from $V_{i_1}, V_{i_2}, V_{i_3}, V_{i_4}$.
            Let $\sigma_c \in \mathcal{S}_c$ such that $\{v_c, v_{c,\sigma_c}\} \in M$.
            We claim that $\sigma_c$ must be consistent with $\sigma$.
            Indeed, if $\sigma_c$ assigns value $k \in [5]$ to $x_i \in V_2$,
            then it follows that $(\setdef{u^{i,j}_k}{k \in [5]} \setminus \{u^{i,j}_{\sigma(x_i, j)}\}) \notin V_M$
            due to the associated XOR gadgets, thus $\{u^{i,j}, u^{i,j}_{\sigma(x_i, j)}\} \in M$.
            On the other hand, if $\sigma$ assigns value $k \in [5]$ to $x_i \in V_2$,
            then
            \begin{itemize}
                \item if $k=1$, then in $\hat{B}_{i,j}$ we have that $b_1,b_2,\chi_2,y_2 \notin V_M$, implying that $\{y_1,a\}, \{\chi_1,l_1\} \in M$,
                \item if $k=2$, then in $\hat{B}_{i,j}$ we have that $b_1,b_2,\chi_1,y_2 \notin V_M$, implying that $\{y_1,a\}, \{\chi_2,l_2\} \in M$,
                \item if $k=3$, then in $\hat{B}_{i,j}$ we have that $a,a',y_1,y_2 \notin V_M$, implying that $\{b_1,b_2\} \in M$,
                \item if $k=4$, then in $\hat{B}_{i,j}$ we have that $b_1,b_2,\chi_2,y_1 \notin V_M$, implying that $\{y_2,a'\}, \{\chi_1,l_1\} \in M$,
                \item if $k=5$, then in $\hat{B}_{i,j}$ we have that $b_1,b_2,\chi_1,y_1 \notin V_M$, implying that $\{y_2,a'\}, \{\chi_1,l_1\} \in M$.
            \end{itemize}
            This completes the proof.
        \end{nestedproof}

        Finally, to bound the pathwidth of the graph $G$,
        start with the decomposition of the primal graph of $\psi$ and
        in each $B_j$ replace each $x_i \in V_2 \cap B_j$ with all the vertices of $\hat{B}_{i,j}$.
        We further add in $B_j$ all the vertices of $\hat{B}_{i,j+1}$ for $x_i \in V_2 \cap B_{j+1}$,
        as well as the vertices $r$ and $r'$.
        For each constraint $c$, let $b(c)=j$ and add into $B_j$ all the
        vertices of $\hat{C}_c$ and any private vertex of a XOR gadget involving a vertex of $\hat{C}_c$,
        for a total of at most $1 + 5^4 + 2\binom{5^4}{2} + 8 \cdot 5^4$ vertices.
        So far each bag contains $p + \bO(\log p)$ vertices.
        To cover the remaining vertices,
        for each $j \in [t]$ we replace the bag $B_j$ with a sequence of bags such that
        all of them contain the vertices we have added to $B_j$ so far,
        the first bag contains the vertices $a$ belonging to all $\hat{B}_{i,j}$ for all $x_i \in V_1 \cap B_j$ and
        the last bag contains the vertices $a'$ belonging to all $\hat{B}_{i,j}$ for all $x_i \in V_1 \cap B_j$.
        We insert a sequence of $\bO(p)$ bags between these two,
        at each step adding all vertices apart from $a$ of a block gadget $\hat{B}_{i,j}$ and then removing
        all vertices of the same block gadget apart from $a'$, for $x_i \in V_1 \cap B_j$.
\end{proof}

\subsection{Parameterization by Clique-width}\label{subsec:acyclic:cw}

Before describing our algorithm, we first define some necessary notions.

% \begin{toappendix}

\subparagraph{Partitions.}
A \emph{partition} $p$ of a set $L$ is a collection of non-empty subsets of $L$
that are pairwise non-intersecting and such that $\bigcup_{p_i \in p} p_i = L$;
each set in $p$ is called a \emph{block} of $p$.
The set of partitions of a finite set $L$ is denoted by $\Pi(L)$,
and $(\Pi(L),\sqsubseteq)$ forms a lattice where $p \sqsubseteq q$
if for each block $p_i$ of $p$ there is a block $q_j$ of $q$ with
$p_i \subseteq q_j$.
The join operation of this lattice is denoted by $\sqcup$.
For example, we have $\{ \{ 1,2\},\{3,4\},\{5\}\} \sqcup \{\{1\},\{2,3\},\{4\},\{5\}\}= \{ \{ 1,2,3,4\}, \{5\}\}$.
Let $\block(p)$ denote the number of blocks of a partition $p$.
Observe that $\varnothing$ is the only partition of the empty set.
A \emph{weighted partition} is an element of $\Pi(L) \times \mathbb{N}$ for some finite set $L$.
For $p \in \Pi(L)$ and $X \subseteq L$, let $p_{\downarrow X} \in \Pi(X)$ be the partition
$\setdef{p_i \cap X}{p_i \in p} \setminus \{\varnothing\}$,
and for a set $Y$, let $p_{\uparrow Y} \in \Pi(L \cup Y)$ be the partition
$p \cup \left(\bigcup_{y \in Y \setminus L} \{\{y\}\}\right)$.

% \end{toappendix}

\begin{theoremrep}[\appsymb]
    There is an algorithm that,
    given an integer $\ell$ as well as a graph $G$ along with an irredundant clique-width expression $\psi$ of $G$ of width $\cw$,
    determines whether $G$ has an acyclic matching of size at least $\ell$ in time $2^{\bO(\cw)} n^{\bO(1)}$.
\end{theoremrep}

\begin{proof}
    The proof follows along the lines of the $2^{\bO(\cw)} n^{\bO(1)}$ algorithm for {\FVS} by Bergougnoux and Kant\'e~\cite{tcs/BergougnouxK19}.
    There, the authors develop a general framework, expanding upon the one of Bodlaender et al.~\cite{iandc/BodlaenderCKN15},
    in order to cope with various problems with connectivity constraints on graphs of bounded clique-width.
    Using their framework, we can perform DP over the clique-width expression and, starting from a partial solution,
    eventually build an optimal one.

    Instead of \AcyclicM, we solve the equivalent problem of asking,
    given a graph $G$ and an integer $k$, whether there exists $S \subseteq V(G)$ of size $|S| \ge k$ such that
    $G[S]$ is a forest that contains a perfect matching.
    Notice that $G$ has an acyclic matching of size $\ell$ if and only if there exists such a set $S$
    of size $|S| \ge 2\ell$.
    In the following let $\mathcal{H}$ denote the set of all $\cw$-labeled graphs generated
    by a subexpression of the given clique-width expression $\psi$ of the input graph $G$.

    For every $H \in \mathcal{H}$ and for every subset $L \subseteq \setdef{i \in [\cw]}{\lab^{-1}_H(i) \neq \varnothing}$,
    we compute a set of weighted partitions $\mathcal{A} \subseteq \Pi(L) \times \mathbb{N}$.
    Each weighted partition $(p,w) \in \mathcal{A} \subseteq \Pi(L) \times \mathbb{N}$ is intended to mean that
    there is a partial solution whose vertices $S \subseteq V(H)$ have weight $w$ and $S \cap \lab^{-1}_H(i) \neq \varnothing$
    for all labels $i \in L$ and $p$ is the transitive closure of the following equivalence relation $\sim$ on $L$:
    $i \sim j$ if there exist an $i$-vertex and a $j$-vertex in the same component of $H[S]$.
    Intuitively, for each label $i$ in $L$, we expect the $i$-vertices of $S$ to have an additional neighbor in any
    extension of $S$ into an optimum solution, therefore, we can consider all vertices of $S \cap \lab^{-1}_H(i)$ as
    one vertex in terms of connectivity.
    On the other hand, the vertices of $S \cap \lab^{-1}_H(j)$, where $j \in [\cw] \setminus L$ and $S$ contains at least one $j$-vertex,
    are expected to have no additional neighbor in any extension of $S$ into an optimum solution.
    Consequently, those vertices no longer play a role in the connectivity of the solution.
    These expectations allow us to represent the connected components of $H[S]$ by $p$.
    Our algorithm will guarantee that the weighted partitions computed from $(p,w)$ are computed accordingly to these expectations.
    It remains to deal with the acyclicity constraint, and to do this,
    we need to certify that whenever we join two weighted partitions and keep the result as a partial solution,
    it does not correspond to a partial solution with cycles.
    To this end, we employ the machinery introduced by~\cite{tcs/BergougnouxK19} and is based on~\cite{iandc/BodlaenderCKN15},
    which, for the sake of completeness, we present in \cref{sec:acyclic:cw:framework}.

    We use the weighted partitions defined in~\cite{tcs/BergougnouxK19} to represent the partial solutions.
    At each step of our algorithm we will ensure that the stored weighted partitions correspond to acyclic partial solutions.
    Since the framework of~\cite{tcs/BergougnouxK19} deals only with \emph{connected} acyclic solutions,
    as in~\cite{tcs/BergougnouxK19}, we introduce a hypothetical new vertex $v_0$ that is universal,
    and we compute a pair $(F,E_0)$ so that $F$ is a maximum induced forest of $G$ that has a perfect matching,
    $E_0$ is a subset of edges incident with $v_0$,
    and $(V(F) \cup \{v_0\}, E(F) \cup E_0)$ is a tree.

    We start with some definitions and notation.
    Let $H \in \mathcal{H}$, and consider a subset of its vertices $S \subseteq V(H)$ and
    a matching $M \subseteq E(H)$ with $V_M \subseteq S$, that is, all the vertices incident with
    edges in $M$ belong to $S$.
    In that case, we say that a vertex $v \in S$ is \emph{unsaturated} in $(S,M)$ if $v \notin V_M$,
    otherwise, i.e., if $v \in V_M$, we say that it is \emph{saturated} in $(S,M)$ instead.
    We say that such a pair $(S,M)$ is a \emph{partial solution} of $H$ if $H[S]$ is acyclic.
    Finally, we say that a partial solution $(S,M)$ of $H$ is \emph{extensible} if there exists an acyclic matching $M^*$ of
    $G$ such that $V_{M^*} \cap V(H) = S$ and $M^* \supseteq M$.

    \begin{claim}\label{claim:acyclic:cw:algorithm}
        Let $(S,M)$ be a partial solution of $H$ that is extensible,
        and let $M^*$ be an acyclic matching of $G$ such that
        $V_{M^*} \cap V(H) = S$ and $M^* \supseteq M$.
        Then, for all $i \in [\cw]$, the following hold.
        \begin{enumerate}
            \item There exists at most one unsaturated vertex of label $i$ in $(S,M)$,
            that is, $|(S \setminus V_M) \cap \lab^{-1}_H(i)| \le 1$.

            \item If $|(S \setminus V_M) \cap \lab^{-1}_H(i)| = 1$,
            then each vertex of $S \cap \lab^{-1}_H(i)$ belongs to a distinct connected component of $H[S]$.

            \item If $|S \cap \lab^{-1}_H(i)| \ge 2$,
            then the vertices of $S \cap \lab^{-1}_H(i)$ take part in at most one join operation
            with a non-empty set,
            that is, there exists at most one graph $H_2$ further in $\psi$
            such that $H_2 = \eta_{i',j'}(H_1)$ with $S \cap \lab^{-1}_H(i) = S \cap \lab^{-1}_{H_1}(i')$
            and $V_{M*} \cap \lab^{-1}_{H_1}(j') \neq \varnothing$.
        \end{enumerate}
    \end{claim}

    \begin{claimproof}
        First notice that for all $H \in \mathcal{H}$ it holds that $H[V_{M^*} \cap V(H)]$ is a subgraph of $G[V_{M^*}]$,
        therefore $H[V_{M^*} \cap V(H)]$ is acyclic.

        Towards a contradiction, let $v_1,v_2 \in (S \setminus V_M) \cap \lab^{-1}_H(i)$,
        and let $u_1,u_2 \in V_{M^*}$ such that $\{v_1,u_1\}, \{v_2,u_2\} \in M^*$.
        There are two cases.
        In the first case, at some point further in the clique-expression there exists a single join operation
        that saturates both $v_1$ and $v_2$, that is,
        there exist $H_1,H_2 \in \mathcal{H}$ and $H_2 = \eta_{i',j'}(H_1)$,
        with $v_1,v_2 \in \lab^{-1}_{H_1}(i')$ and $u_1,u_2 \in \lab^{-1}_{H_1}(j')$.
        Then however $H_2[S]$ contains a cycle, which contradicts the acyclicity of $H_2[V_{M^*} \cap V(H_2)]$
        since $S \subseteq V_{M^*} \cap V(H_2)$.
        In the second case, $v_1$ and $v_2$ are saturated by different join operations.
        This means however that after the first one, they are in the same connected component,
        thus the second join operation results in a cycle, again leading to a contradiction.
        The argument for the third item of the statement is analogous.

        Similarly, one can show that if $|(S \setminus V_M) \cap \lab^{-1}_H(i)| = 1$
        and there exist two vertices of $S \cap \lab^{-1}_H(i)$ belonging to the same connected component of $H[S]$,
        then the join operation that saturates $v$ results in a cycle, thus leading to a contradiction.
    \end{claimproof}

    We proceed by dynamic programming, aiming to construct in a bottom-up fashion all extensible partial solutions.
    In order to reduce the sizes of the sets stored in the entries of the DP tables,
    we will express the steps of the algorithm in terms of the operators on weighted partitions defined in~\cite{tcs/BergougnouxK19}.
    Let $H \in \mathcal{H}$.
    We are interested in storing \emph{ac-representative sets} of all partial solutions $(S,M)$ of $H$ that may produce a solution,
    thus we assume that the stored partial solutions satisfy the properties of \cref{claim:acyclic:cw:algorithm}.
    Consequently, it holds that if $S$ contains at least $2$ $i$-vertices, then the vertices of this label may
    partake in \emph{at most one} join operation with a non-empty set further in the clique-expression.
    Let $J \subseteq [\cw]$ denote the set of all labels from which $S$ contains at least $2$ vertices,
    and consider a partition $J_1 \uplus J_2 = J$ such that for the labels in $J_1$ there exists further in the clique-expression
    a join operation with a non-empty set while for those in $J_2$ it does not exist.
    Notice that in terms of connectivity, we can treat all vertices of $S \cap \lab^{-1}_H(i)$ for $i \in J_1$
    as one vertex, as they will all eventually end up in the same connected component of the final solution.
    On the other hand, for the vertices labeled $j \in J_2$,
    we can take into account how they affect the connectivity of the graph and subsequently ignore them,
    as no other edges will be added incident with them.
    We remark that if $J_2 = \varnothing$,
    then we can encode the connectivity of the graph by storing the partition $p \in \Pi([\cw])$ where
    $i$ and $j$ are in the same block if there are
    an $i$-vertex and a $j$-vertex in the same connected component of $H[S]$.

    Now consider the case where $H_1,H_2 \in \mathcal{H}$ with $H_2 = \eta_{i,j}(H_1)$ and $(S,M)$ is a partial solution of $H_1$,
    where $i,j \in J \subseteq [\cw]$.
    Notice that this join operation might result in cycles forming in $H_2[S]$, e.g.,
    whenever an $i$-vertex and a $j$-vertex are non-adjacent and belong to the same connected component of $H_1[S]$,
    or the number of $i$-vertices and $j$-vertices are both at least $2$ in $S$.
    Unfortunately, we are not able to handle all these cases with the operators on weighted partitions.
    To resolve the situation where an $i$-vertex and a $j$-vertex are already adjacent,
    we consider \emph{irredundant} $\cw$-expressions, i.e., whenever an operation $\eta_{i,j}$ is used there are no edges
    between $i$-vertices and $j$-vertices.
    For the other cases, we index the DP tables with total functions $s \colon [\cw] \to \Gamma$ called \emph{signatures}, that indicate, for each label $i \in [\cw]$, whether $\lab_H^{-1}(i)$ intersects $S$ and $V_M$,
    with $\Gamma = \{\gamma_0,\gamma_{\tilde{1}},\gamma_1,\gamma_{\tilde{2}},\gamma_2,\gamma_{-2}\}$.
    In particular, we say that a partial solution $(S,M)$ of $H$ is \emph{compatible} with the signature $s$ if for all $i \in [\cw]$
    it holds that
    \begin{itemize}
        \item if $s(i) = \gamma_0$, then $S \cap \lab^{-1}_H(i) = \varnothing$,

        \item if $s(i) = \gamma_{\tilde{1}}$, then $S \cap \lab^{-1}_H(i) = \{v\}$ and $v \notin V_M$,

        \item if $s(i) = \gamma_1$, then $S \cap \lab^{-1}_H(i) = \{v\}$ and $v \in V_M$,

        \item if $s(i) = \gamma_{\tilde{2}}$, then $|S \cap \lab^{-1}_H(i)| \ge 2$ and $|(S \setminus V_M) \cap \lab^{-1}_H(i)| = 1$,

        \item if $s(i) \in \{\gamma_2, \, \gamma_{-2}\}$, then $|S \cap \lab^{-1}_H(i)| \ge 2$ and $(S \setminus V_M) \cap \lab^{-1}_H(i) = \varnothing$.
    \end{itemize}
    Defining the signatures like so allows us to resolve the problem with the sets $J_1$ and $J_2$ by forcing each label
    $i$ with $s(i) = \gamma_2$ to wait for \emph{exactly} one clique-width operation $\eta_{i,\ell}$ for some $\ell \in [\cw]$ with $s(\ell) \neq \gamma_0$,
    while any label $j$ with $s(j) = \gamma_{-2}$ is forbidden to partake in any such operation.
    Notice that by definition, any label class $i$ with $s(i) = \gamma_{\tilde{2}}$ is also forced to wait for exactly $1$ such join operation,
    as no vertex is unsaturated in $G$.
    Thus, we translate all the acyclicity tests to the $\acjoin$ operation.
    The following notion of \emph{certificate graph} formalizes this requirement,
    where the vertices in $V^+_s$ represent the expected future neighbors of all the vertices in
    $S \cap \lab^{-1}_H(s^{-1}(\{\gamma_{\tilde{2}},\gamma_2\}))$.

    \begin{definition}[Certificate graph of a solution]\label{defn:certif}
        Let $H \in \mathcal{H}$, $F$ an induced forest of $H$,
        $s \colon [\cw] \to \Gamma$,
        and $E_0$ a subset of edges incident with $v_0$.
        Let $V^+_s = \setdef{v^+_i}{i \in s^{-1}(\{\gamma_{\tilde{2}},\gamma_2\})}$.
        The \emph{certificate graph of $(F,E_0)$ with respect to $s$},
        denoted by $\CG(F,E_0,s)$,
        is the graph $(V(F) \cup V_s^+ \cup \{v_0\}, E(F) \cup E_0 \cup E_s^+)$ with
        \[
            E_s^+ = \bigcup_{i \in s^{-1}(\{\gamma_{\tilde{2}},\gamma_2\})} \setdef{\{v,v^+_i\}}{v \in V(F) \cap \lab^{-1}_H(i)}.
        \]
    \end{definition}

    We are now ready to define the sets of weighted partitions whose representatives we manipulate in our dynamic programming tables.

    \begin{definition}[{Weighted partitions in $\mathcal{A}_H[s]$}]\label{defn:tabfvs}
        Let $H \in \mathcal{H}$ and consider a signature $s \colon [\cw] \to \Gamma$.
        The entries of $\mathcal{A}_H[s]$ are all weighted partitions
        $(p,w) \in \Pi(s^{-1}(\{\gamma_{\tilde{1}},\gamma_1,\gamma_{\tilde{2}},\gamma_2\}) \cup \{v_0\}) \times \mathbb{N}$
        such that there exist a partial solution $(S,M)$ of $H$ and $E_0 \subseteq \setdef{\{v_0 v\}}{v \in S}$
        so that $\wc (S) = w$, and
        \begin{enumerate}
            \item $(S, M)$ is compatible with $s$,

            \item the certificate graph $\CG(H[S],E_0,s)$ is a forest,

            \item each connected component of $\CG(H[S],E_0,s)$ has at least one vertex in
            $\lab^{-1}_H(s^{-1}(\{\gamma_{\tilde{1}},\gamma_1,\gamma_{\tilde{2}},\gamma_2\})) \cup \{v_0\}$,

            \item the partition $p$ equals $(s^{-1}(\{\gamma_{\tilde{1}},\gamma_1,\gamma_{\tilde{2}},\gamma_2\}) \cup \{v_0\})/\sim$,
            where $i \sim j$ if and only if a vertex in $S \cap \lab^{-1}_H(i)$ is connected, in $\CG(H[S],E_0,s)$,
            to a vertex in $S \cap \lab^{-1}_H(j)$; we consider $\lab^{-1}_H(v_0)=\{v_0\}$.
        \end{enumerate}
    \end{definition}
    Conditions (2) and (4) guarantee that $(S \cup \{v_0\},E(H[S]) \cup E_0)$ can be extended into a tree, if any.
    They also guarantee that cycles detected through the $\acjoin$ operation correspond to cycles,
    and each cycle can be detected with it.
    In the following we call any quadruple $(S,M,E_0,(p,\wc(S)))$ a \emph{candidate solution} in $\mathcal{A}_H[s]$ if Condition (4) is satisfied,
    and if in addition Conditions (1)-(3) are satisfied, we call it a \emph{solution} in $\mathcal{A}_H[s]$.

    In that case, the size of a maximum acyclic matching of $G$ corresponds to the maximum,
    over all signatures $s \colon [\cw] \to \Gamma$
    with $s^{-1}(\{\gamma_{\tilde{1}},\gamma_{\tilde{2}},\gamma_{2}\}) = \varnothing$,
    of $\max \setdef{w}{(\{ s^{-1}(\gamma_{1}) \cup \{v_0\} \},w) \in \mathcal{A}_G[s]}$.
    Indeed, by definition, if $(\{ s^{-1}(\gamma_{1}) \cup \{v_0\} \},w)$ belongs to $\mathcal{A}_{G}[s]$,
    then there exists a partial solution $(S,M)$ of $G$ with $V_M=S$ and $\wc(S)=w$,
    as well as a set $E_0$ of edges incident with $v_0$ such that $(S \cup \{v_0\},E(H[S]) \cup E_0)$ is a tree.
    This follows from the fact that if $s^{-1}(\{\gamma_{\tilde{2}},\gamma_{2}\}) = \varnothing$,
    then we have $\CG(H[S],E_0,s) = (S \cup \{v_0\}, E(H[S]) \cup E_0)$.

    Our algorithm will store, for each $H \in \mathcal{H}$ and each $s \colon [\cw] \to \Gamma$,
    an ac-representative set $\DPt_H[s]$ of $\mathcal{A}_H[s]$.
    We are now ready to give the different steps of the algorithm, depending on the clique-width operations.
    Recall that $f[\alpha \mapsto \beta]$ denotes the function obtained from $f$
    that maps $\alpha$ to $\beta$ instead of $f(\alpha)$.

    \proofsubparagraph{Singleton $H = i(v)$.}
    For $H = i(v)$, notice that $\lab^{-1}_H(i) = \{v\}$ and $\lab^{-1}_H(w) = \varnothing$ for all $w \in [\cw] \setminus \{i\}$.
    Consequently, for $s \colon [\cw] \to \Gamma$ we set the following values,
    where $s_0$ is the signature such that $s_0(w)=\gamma_0$ for all $w \in [\cw]$.
    \[
        \DPt_H[s] =
            \begin{cases}
                \bigl\{ \bigl( \{\{v_0\} \} ,0 \bigr) \bigr\}                                                   &\text{if $s = s_0$,}\\
                \bigl\{ \bigl( \{\{i,v_0\}\}, \wc(v) \bigr), \bigl( \{\{i\},\{v_0\}\}, \wc(v) \bigr) \bigr\}    &\text{if $s = s_0 [i \mapsto \gamma_{\tilde{1}}]$,}\\
                \varnothing                                                                                     &\text{otherwise.}
            \end{cases}
    \]
    Since $|V(H)|=1$, there is no partial solution $(S,M)$ intersecting $\lab^{-1}_H(i)$ on at least two vertices while $H$ has no edges,
    so the set of weighted partitions satisfying \cref{defn:tabfvs} equals the empty set for $s(i) \neq \{\gamma_0,\gamma_{\tilde{1}}\}$.
    If $s(i) = \gamma_{\tilde{1}}$, there are two possibilities, depending on whether $E_0 = \varnothing$ or
    $E_0 = \{v v_0\}$.
    We can thus conclude that $\DPt_H[s] = \mathcal{A}_H[s]$ is correctly computed.

    \proofsubparagraph{Joining labels with edges, $H = \eta_{i,j}(H')$.}
    Before describing how to populate the table in this case,
    we first define a helper partial function%
    \footnote{We omit the use of braces in case of singletons for the sake of readability.}
    $f \colon \Gamma \times \Gamma \to 2^{\Gamma \times \Gamma}$
    as follows, where the row denotes the first input and the column the second:
    \[
        \begin{array}{r|cccccc}
            f 	                & \gamma_0 							& \gamma_{\tilde{1}} 																			& \gamma_1 							& \gamma_{\tilde{2}} 				& \gamma_2 							& \gamma_{-2} \\
            \hline
            \gamma_0 			& (\gamma_0, \gamma_0) 				& (\gamma_0, \gamma_{\tilde{1}}) 																& (\gamma_0, \gamma_1) 				& (\gamma_0, \gamma_{\tilde{2}}) 	& (\gamma_0, \gamma_2) 				& (\gamma_0, \gamma_{-2}) \\
            \gamma_{\tilde{1}} 	& (\gamma_{\tilde{1}}, \gamma_0) 	& \left\{\makecell{(\gamma_{\tilde{1}}, \gamma_{\tilde{1}}) \\ (\gamma_1,\gamma_1)} \right\} 	& (\gamma_{\tilde{1}}, \gamma_1) 	& (\gamma_1, \gamma_{-2}) 			& (\gamma_{\tilde{1}}, \gamma_{-2}) & \varnothing \\
            \gamma_1 			& (\gamma_1, \gamma_0) 				& (\gamma_1, \gamma_{\tilde{1}}) 																& (\gamma_1, \gamma_1) 				& \varnothing 								& (\gamma_1, \gamma_{-2}) 	& \varnothing \\
            \gamma_{\tilde{2}} 	& (\gamma_{\tilde{2}}, \gamma_0)	& (\gamma_{-2}, \gamma_1) 																		& \varnothing 						& \varnothing 								& \varnothing 				& \varnothing \\
            \gamma_2 			& (\gamma_2, \gamma_0) 				& (\gamma_{-2}, \gamma_{\tilde{1}}) 															& (\gamma_{-2}, \gamma_1) 			& \varnothing 								& \varnothing 				& \varnothing \\
            \gamma_{-2}			& (\gamma_{-2}, \gamma_0) 			& \varnothing																					& \varnothing 						& \varnothing 								& \varnothing 				& \varnothing
        \end{array}
    \]
    In the following, let $f^{-1}(\alpha, \beta) = \setdef{(\alpha', \beta') \in \Gamma \times \Gamma}{(\alpha,\beta) \in f(\alpha', \beta')}$.
    For $s \colon [\cw] \to \Gamma$ we set
    \[
        \DPt_H[s] =
            \begin{cases}
                \DPt_{H'}[s]		            &\text{if $\gamma_0 \in \{s(i), \, s(j)\}$,}\\
                \varnothing			            &\text{if $f^{-1}(s(i),s(j)) = \varnothing$,}\\
                \acreduce ( \rmc(\mathcal{A}) ) &\text{otherwise,}
            \end{cases}
    \]
    where
    \[
        \mathcal{A} = \bigcup_{(\alpha',\beta') \in f^{-1}(s(i),s(j))}
            \proj \left(
                \acjoin\left( \DPt_{H'} \Bigl[ s[i \mapsto \alpha'][j \mapsto \beta'] \Bigr], \,
                \{ ( \{\{i,j\}\},0)\}\right), \,
                s^{-1}(\gamma_{-2}) \cap \{i,j\}
            \right).
    \]

    Notice that in the first case we do not need to use the operators $\acreduce$ and $\rmc$ since we update
    $\DPt_H[s]$ with one table from $\DPt_{H'}$.
    The second case can only be when (i) $s(i)=s(j)=\gamma_{-2}$,
    or (ii) $\gamma_0 \notin \{s(i),s(j)\}$ and
    $\{\gamma_{\tilde{2}},\gamma_{2}\} \cap \{s(i),s(j)\} \neq \varnothing$.
    As for the third case, we have that $s(i),s(j) \in \{\gamma_{\tilde{1}}, \gamma_{1}, \gamma_{-2}\}$ and
    either $s(i) \neq \gamma_{-2}$ or $s(j) \neq \gamma_{-2}$.
    Intuitively, we consider the weighted partitions $(p,w) \in \mathcal{A}$ such that
    $i$ and $j$ belong to different blocks of $p$,
    we merge the blocks containing $i$ and $j$,
    remove the elements in $s^{-1}(\gamma_{-2}) \cap \{i,j\}$ from the resulting block,
    and add the resulting weighted partition to $\DPt_H[s]$.
    Notice that we also have to consider whether a new edge has been added to the partial solution due to the join operation.

    \proofsubparagraph{Relabeling, $H = \rho_{i \to j}(H')$.}
    As before, we first define a helper partial function
    $g \colon \Gamma \times \Gamma \to 2^{\Gamma}$
    as follows, where the row denotes the first input and the column the second:
    \[
        \begin{array}{r|cccccc}
            g 					& \gamma_0 				& \gamma_{\tilde{1}} 	& \gamma_1 												& \gamma_{\tilde{2}} 	& \gamma_2 				& \gamma_{-2} \\
            \hline
            \gamma_0 			& \gamma_0 				& \gamma_{\tilde{1}} 	& \gamma_1 												& \gamma_{\tilde{2}} 	& \gamma_2 				& \gamma_{-2} \\
            \gamma_{\tilde{1}} 	& \gamma_{\tilde{1}} 	& \varnothing 			& \gamma_{\tilde{2}} 									& \varnothing 			& \gamma_{\tilde{2}} 	& \varnothing \\
            \gamma_1 			& \gamma_1				& \gamma_{\tilde{2}}	& \left\{ \makecell{ \gamma_2 \\ \gamma_{-2}} \right\}	& \gamma_{\tilde{2}} 	& \gamma_2				& \gamma_{-2} \\
            \gamma_{\tilde{2}} 	& \gamma_{\tilde{2}}	& \varnothing 			& \gamma_{\tilde{2}} 									& \varnothing 			& \gamma_{\tilde{2}} 	& \varnothing \\
            \gamma_2 			& \gamma_2 				& \gamma_{\tilde{2}} 	& \gamma_2 												& \gamma_{\tilde{2}}	& \gamma_2 				& \varnothing \\
            \gamma_{-2}			& \gamma_{-2} 			& \varnothing			& \gamma_{-2} 											& \varnothing 			& \varnothing 			& \gamma_{-2}
        \end{array}
    \]
    In the following, let $\gminusrest(\beta) = \setdef{(\alpha', \beta') \in \Gamma \times \Gamma}{\beta \in f(\alpha', \beta') \text{ and } \gamma_0 \notin \{ \alpha', \beta' \}}$.
    For $s \colon [\cw] \to \Gamma$ we set
    \[
        \DPt_H[s] =
            \begin{cases}
                \acreduce ( \rmc(\mathcal{A}_1 \cup \mathcal{A}_2 \cup \mathcal{A}_3 ) )    &\text{if $s(i) = \gamma_0$,}\\
                \varnothing                                                                 &\text{otherwise},
            \end{cases}
    \]
    where we obtain $\mathcal{A}_1,\mathcal{A}_2,\mathcal{A}_3$ as follows.

    $\mathcal{A}_1$ contains all weighted partitions corresponding to partial solutions not intersecting $\lab^{-1}_{H'}(i)$,
    which are trivially partial solutions for $H$ as well, and we have
    \[
        \mathcal{A}_1 = \DPt_{H'} \Bigl[ s[i \mapsto \gamma_0] \Bigr].
    \]

    $\mathcal{A}_2$ contains all weighted partitions corresponding to partial solutions that
    intersect $\lab^{-1}_{H'}(i)$ but not $\lab^{-1}_{H'}(j)$.
    Notice that the case where a partial solution intersects neither,
    i.e., $s(i) = s(j) = \gamma_0$, is already covered by $\mathcal{A}_1$.
    We set
    \[
        \mathcal{A}_2 =
            \begin{cases}
                \varnothing         &\text{if $s(j) = \gamma_0$,}\\
                \DPt_{H'}[s_{H'}]        &\text{if $s(j) = \gamma_{-2}$},\\
                \proj\big(
                    \acjoin( \DPt_{H'}[s_{H'}], \,
                        \{ ( \{\{i,j\}\},0) \}), \,
                    \{i\}
                \big)               &\text{otherwise},
            \end{cases}
    \]
    where $s_{H'} = s[i \mapsto s(j)][j \mapsto \gamma_0]$.

    The set $\mathcal{A}_3$ has to do with the cases where the partial solution intersects
    both $\lab^{-1}_{H'}(i)$ and $\lab^{-1}_{H'}(j)$, which implies that $s(j) \in \{ \gamma_{\tilde{2}}, \gamma_2, \gamma_{-2}\}$.
    To this end, we consider the following subcases.
    If $s(j) = \gamma_{-2}$, then
    \[
        \mathcal{A}_3 = \bigcup_{(\alpha',\beta') \in \gminusrest(\gamma_{-2})}
                \proj\Big( \DPt_{H'} \Bigl[ s[i \mapsto \alpha'][j \mapsto \beta'] \Bigr] , \, \{i,j\} \Big).
    \]
    On the other hand, if $s(j) \in \{ \gamma_{\tilde{2}}, \gamma_2 \}$, then
    \[
        \mathcal{A}_3 = \bigcup_{(\alpha',\beta') \in \gminusrest(s(j))}
                \proj\bigg(
                    \acjoin \Big(
                        \DPt_{H'} \Bigl[ s[i \mapsto \alpha'][j \mapsto \beta'] \Bigr], \,
                        \{ ( \{\{i,j\}\},0) \} \Big), \,
                    \{i\}
                \bigg).
    \]

    \proofsubparagraph{Disjoint union, $H = H_1 \oplus H_2$.}
    Consider the signatures $s, s_1, s_2$.
    We say that $s_1,s_2$ \emph{agree on $s$} if for all labels $i \in [\cw]$,
    it holds that $s(i) \in g(s_1(i), s_2(i))$.
    In that case, for $s \colon [\cw] \to \Gamma$ we set $\DPt_H[s] = \acreduce(\rmc(\mathcal{A}))$,
    where
    \[
        \mathcal{A} = \bigcup_{s_1,s_2 \text{ agree on } s}
            \acjoin \bigg(
                \proj \Big( \DPt_{H_1}[s_1], \, s^{-1}(\gamma_{-2}) \Big), \,
                \proj \Big( \DPt_{H_2}[s_2], \, s^{-1}(\gamma_{-2}) \Big)
            \bigg).
    \]
    As in~\cite{tcs/BergougnouxK19},
    we need to do the projections before the $\acjoin$ operation.

    The correctness of the algorithm follows along the lines of the {\FVS} algorithm in~\cite{tcs/BergougnouxK19}.
\end{proof}

\begin{proofsketch}
    The proof follows along the lines of the $2^{\bO(\cw)} n^{\bO(1)}$ algorithm for {\FVS} by Bergougnoux and Kant\'e~\cite{tcs/BergougnouxK19}.
    There, the authors develop a general framework, expanding upon the one of Bodlaender et al.~\cite{iandc/BodlaenderCKN15},
    in order to cope with various problems with connectivity constraints on graphs of bounded clique-width.
    Using their framework, we can perform DP over the clique-width expression and, starting from a partial solution,
    eventually build an optimal one.

    Instead of \AcyclicM, we solve the equivalent problem of asking,
    given a graph $G$ and an integer $k$, whether there exists $S \subseteq V(G)$ of size $|S| \ge k$ such that
    $G[S]$ is a forest that contains a perfect matching.
    Notice that $G$ has an acyclic matching of size $\ell$ if and only if there exists such a set $S$
    of size $|S| \ge 2\ell$.
    In the following let $\mathcal{H}$ denote the set of all $\cw$-labeled graphs generated
    by subexpressions of the given clique-width expression $\psi$ of the input graph $G$.

    In order to deal with \FVS, in~\cite{tcs/BergougnouxK19} the authors solve the \textsc{Maximum Induced Tree} problem and
    introduce a \emph{signature} on the partial solutions,
    that encodes (i) the number of vertices that a partial solution contains from each label class ($0$, $1$, or at least $2$),
    and (ii) in case a solution contains at least $2$ vertices of the same label class,
    whether this label class will partake in a join operation with a non-empty set in the future;
    as a matter of fact, the number of such join operations can be at most $1$, as otherwise a cycle is formed.
    For our problem, we have that a partial solution of $H \in \mathcal{H}$
    is a pair $(S,M)$ such that $M$ is a matching of $H$ with $V_M \subseteq S \subseteq V(H)$.
    In that case, the signature of a partial solution is, for each label $i \in [\cw]$,
    whether $S \cap \lab^{-1}_H(i)$ is empty, a singleton, or contains at least $2$ vertices,
    and whether $(S \setminus V_M) \cap \lab^{-1}_H(i)$ is empty or not.
    We notice that in fact $|(S \setminus V_M) \cap \lab^{-1}_H(i)| \le 1$, thus leading to $6$ different states
    $\Gamma = \{\gamma_0, \gamma_{\tilde{1}}, \gamma_1, \gamma_{\tilde{2}}, \gamma_2, \gamma_{-2}\}$
    per label class, with the index denoting the number of vertices in $S \cap \lab^{-1}_H(i)$,
    the tilde denoting that $|(S \setminus V_M) \cap \lab^{-1}_H(i)| = 1$,
    and $\gamma_2$ (resp., $\gamma_{-2}$) denoting that the label class will (resp., will not)
    partake in a join operation in the future.
    This, along with the machinery of the framework of~\cite{tcs/BergougnouxK19},
    allows us to detect any possible cycles as well as forcing the final solution to be a forest containing a perfect matching.
\end{proofsketch}

\section{Conclusion}\label{sec:conclusion} \nosectionappendix
In this paper we studied the {\InducedM} and {\AcyclicM} problems under the perspective of fine-grained structural parameterized complexity.
In particular, we presented tight lower bounds under the pw-SETH for the pathwidth parameterization of both problems
(equivalence even in the case of \InducedM), rendering known treewidth algorithms optimal.
Furthermore, we considered the parameterization by clique-width and developed a tight $\sO(3^{\cw})$ algorithm for \InducedM,
as well as an $2^{\bO(\cw)} n^{\bO(1)}$ algorithm for \AcyclicM, which is optimal under the ETH.
As a direction for future work, it would be interesting to develop an algorithm for {\AcyclicM} parameterized by clique-width
that is optimal under the (pw-)SETH.
Although recently there has been a flurry of results regarding optimal algorithms for connectivity problems parameterized by clique-width~\cite{arxiv/BojikianK24,icalp/BojikianK24,esa/HegerfeldK23},
the case of \FVS, which is closely related to \AcyclicM, remains open.
Another research direction would be to develop an optimal algorithm for {\AcyclicM} parameterized by the cutwidth of the input graph;
such SETH-tight algorithms are known for a plethora of graph problems by now~\cite{arxiv/BojikianCK25,stacs/GroenlandMNS22,tcs/JansenN19,icalp/MarxSS21,jgaa/GeffenJKM20}, including \FVS~\cite{stacs/BojikianCHK23}.

\begin{toappendix}

\section{Framework}\label{sec:acyclic:cw:framework}

For the sake of completeness, here we present the framework introduced in~\cite[Section~3]{tcs/BergougnouxK19}.

The following operation determines whether a join operation between pairs of partitions produces any cycles.

\begin{definition}
    Let $L$ be a finite set.
    We define $\acy \colon \Pi(L) \times \Pi(L) \to \{\true,\false\}$
    such that $\acy(p,q) = \true$ if and only if
    $|L| + \block(p \sqcup q) - (\block(p)+\block(q))=0$.
\end{definition}

As observed in~\cite{tcs/BergougnouxK19}, if $F_p=(L,E_p)$ and $F_q=(L,E_q)$ are forests with components
$p=\cc(F_p)$ and $q=\cc(F_q)$ respectively,
then $\acy(p,q) = \true$ if and only if $E_p \cap E_q = \varnothing$ and $(L,E_p \uplus E_q)$ is a forest.
Furthermore, the following hold.

\begin{lemma}[{\cite[Fact~3.2]{tcs/BergougnouxK19}}]\label{fact:par-rep2}
    Let $L$ be a finite set.
    For all partitions $p,q,r \in \Pi(L)$, it holds that
    $\acy(p,q) \land \acy(p \sqcup q,r) \iff
    \acy(q,r) \land \acy(p,q\sqcup r)$.
\end{lemma}

\begin{lemma}[{\cite[Fact~3.3]{tcs/BergougnouxK19}}]\label{fact:join}
    Let $L$ be a finite set.
    Let $q \in \Pi(L)$ and let $X \subseteq L$ such that no subset of $X$ is a block of $q$.
    Then, for each $p \in \Pi(L \setminus X)$,
    the following hold
    \begin{align}
        p_{\uparrow X} \sqcup q = \{L\} &\iff
        p \sqcup q_{\downarrow (L \setminus X)}
        = \{L \setminus X\},\\
        \acy(p_{\uparrow X},q)  &\iff \acy(p,q_{\downarrow (L \setminus X)}).
    \end{align}
\end{lemma}

We modify in this section the operators on weighted partitions defined in~\cite{iandc/BodlaenderCKN15} in order to express our
dynamic programming algorithms in terms of these operators, and also to deal with acyclicity.

\subparagraph{Rmc.}
Let $L$ be a finite set.
For $\mathcal{A} \subseteq \Pi(L) \times \mathbb{N}$, let
\[
    \rmc(\mathcal{A}) =
        \setdef{(p,w) \in \mathcal{A}}{\forall(p,w') \in \mathcal{A}, \, w' \leq w}.
\]
This operator, defined in~\cite{iandc/BodlaenderCKN15},
is used to remove all the partial solutions whose weights are not maximum with respect to their partitions.

\subparagraph{Ac-Join.}
Let $L'$ be a finite set.
For $\mathcal{A} \subseteq \Pi(L) \times \mathbb{N}$ and $\mathcal{B} \subseteq \Pi(L') \times \mathbb{N}$,
we define $\acjoin(\mathcal{A},\mathcal{B}) \subseteq \Pi(L \cup L') \times \mathbb{N}$ as
% {\small
\[
    \acjoin(\mathcal{A},\mathcal{B}) =
        \setdef{(p_{\uparrow L'}\sqcup q_{\uparrow L},w_1+w_2)}
            {(p,w_1) \in \mathcal{A}, \, (q,w_2) \in \mathcal{B}, \textrm{ and }
            \acy(p_{\uparrow L'},q_{\uparrow L})}.
\]%}
This operator is more or less the same as the one in~\cite{iandc/BodlaenderCKN15},
except that we incorporate the acyclicity condition.
It is used to construct partial solutions while guaranteeing the acyclicity.

\subparagraph{Project.}
For $X \subseteq L$ and $\mathcal{A} \subseteq \Pi(L) \times \mathbb{N}$,
let $\proj(\mathcal{A},X) \subseteq \Pi(L \setminus X) \times \mathbb{N}$ be
\[
    \proj(\mathcal{A},X) = \setdef{(p_{\downarrow (L \setminus X)},w)}
        {(p,w) \in \mathcal{A} \textrm{ and } \forall p_i \in p, \, (p_i \setminus X) \neq \varnothing}.
\]
This operator considers all the partitions such that no block is completely contained in $X$,
and then removes $X$ from those partitions.
We index our dynamic programming tables with functions that inform on the label classes playing
a role in the connectivity of partial solutions, and this operator is used to remove from the
partitions the label classes that are required to no longer play a role in the connectivity of
the partial solutions.
If a partition has a block fully contained in $X$, it means that this block will remain
disconnected in the future steps of our dynamic programming algorithm, and that is why
we remove such partitions (besides those with cycles).

\bigskip

One needs to perform the above operations efficiently, and this is guaranteed by the following,
which assumes that $\log(|\mathcal{A}|) \le |L|^{\bO(1)}$ for each
$\mathcal{A} \subseteq \Pi(L) \times \mathbb{N}$ (this can be established by applying the operator $\rmc$).

\begin{proposition}[{\cite[Proposition~3.4]{tcs/BergougnouxK19}}]\label{prop:op}
    The operator $\acjoin$ can be performed in time $|\mathcal{A}| \cdot |\mathcal{B}| \cdot |L \cup L'|^{\bO(1)}$
    and the size of its output is upper-bounded by $|\mathcal{A}| \cdot |\mathcal{B}|$.
    The operators $\rmc$ and $\proj$ can be performed in time $|\mathcal{A}| \cdot |L|^{\bO(1)}$,
    and the sizes of their outputs are upper-bounded by $|\mathcal{A}|$.
\end{proposition}

We now define the notion of representative sets of weighted partitions which is the same as the one in~\cite{iandc/BodlaenderCKN15},
except that we need to incorporate the acyclicity condition as for the $\acjoin$ operator above.

\begin{definition}\label{defn:par-rep}
    Let $L$ be a finite set and let $\mathcal{A} \subseteq \Pi(L) \times \mathbb{N}$.
    For $q \in \Pi(L)$, let
    \[
        \acopt(\mathcal{A},q) = \max \setdef{w}{(p,w) \in \mathcal{A}, \, p \sqcup q = \{L\}, \textrm{ and } \acy(p,q)}.
    \]
    A set of weighted partitions $\mathcal{A}' \subseteq \Pi(L) \times \mathbb{N}$ \emph{ac-represents} $\mathcal{A}$ if
    for each $q \in \Pi(L)$, it holds that $\acopt(\mathcal{A},q)=\acopt(\mathcal{A}',q)$.

    Let $Z$ and $L'$ be two finite sets.
    A function $f \colon 2^{\Pi(L) \times \mathbb{N}} \times Z \to 2^{\Pi(L') \times \mathbb{N}}$ is said to
    \emph{preserve ac-representation} if for each
    $\mathcal{A},\mathcal{A}' \subseteq \Pi(L) \times \mathbb{N}$ and $z \in Z$,
    it holds that $f(\mathcal{A}',z)$ ac-represents $f(\mathcal{A},z)$ whenever $\mathcal{A}'$ ac-represents $\mathcal{A}$.
\end{definition}

At each step of our algorithm, we will compute a small set $\mathcal{S}'$ that ac-represents the set $\mathcal{S}$ containing all the partial solutions.
In order to prove that we compute an ac-representative set of $\mathcal{S}$,
we show that $\mathcal{S} = f(\mathcal{R}_1,\ldots,\mathcal{R}_t)$ with $f$ a composition of functions that preserve ac-representation,
and $\mathcal{R}_1,\ldots,\mathcal{R}_t$ the sets of partials solutions associated with the previous steps of the algorithm.
To compute $\mathcal{S}'$, it is sufficient to compute $f(\mathcal{R}'_1,\ldots,\mathcal{R}_t')$,
where each $\mathcal{R}_i'$ is an ac-representative set of $\mathcal{R}_i$.
The following lemma guarantees that the operators we use preserve ac-representation.

\begin{lemma}[{\cite[Lemma~3.6]{tcs/BergougnouxK19}}]\label{lemma:acyclic:cw:framework:ac_preserve}
    The union of two sets in $2^{\Pi(L) \times \mathbb{N}}$ and the operators
    $\rmc$, $\proj$, and $\acjoin$ preserve ac-representation.
\end{lemma}

Lastly, it holds that, for every set $\mathcal{A} \subseteq \Pi(L) \times \mathbb{N}$, we can find,
in time $|\mathcal{A}| \cdot 2^{\bO(|L|)}$, a subset $\mathcal{A}' \subseteq \mathcal{A}$ of size at
most $|L| \cdot 2^{|L|}$ that ac-represents $\mathcal{A}$.
The constant $\omega$ denotes the matrix multiplication exponent.

\begin{theorem}[{\cite[Theorem~3.8]{tcs/BergougnouxK19}}]\label{thm:reduce1}
    There exists an algorithm $\acreduce$ that,
    given a set of weighted partitions $\mathcal{A} \subseteq \Pi(L) \times \mathbb{N}$,
    outputs in time $|\mathcal{A}| \cdot 2^{(\omega-1) \cdot |L|} \cdot |L|^{\bO(1)}$
    a subset $\mathcal{A}'$ of $\mathcal{A}$ that ac-represents $\mathcal{A}$ and
    such that $|\mathcal{A}'| \leq |L| \cdot 2^{|L|-1}$.
\end{theorem}

\end{toappendix}

\bibliography{bibliography}

\end{document}